\documentclass[10pt]{article}
 
\usepackage[utf8]{inputenc}

\usepackage{amsmath, amssymb, amsthm}
\usepackage{graphicx}
\usepackage{wrapfig}
\usepackage{comment}
\usepackage{color}
\definecolor{lightred}{rgb}{1, .30, 0.30}
\definecolor{lightblue}{rgb}{0.30, .30, 1.0}

\providecommand{\keywords}[1]{\noindent\textbf{\textit{Keywords: }} #1}

\newtheorem{theorem}{Theorem}
\newtheorem{definition}{Definition}
\newtheorem{lemma}{Lemma}
\newtheorem{proposition}{Proposition}
\newtheorem{corollary}{Corollary}
\theoremstyle{remark}
\newtheorem{rmrk}{Remark}

\newcommand{\poly}{\mathrm{poly}}

\begin{document}

\title{The expressiveness  of quasiperiodic and minimal shifts of finite type\footnote{Supported by the ANR grant Racaf  ANR-15-CE40-0016-01. 
Preliminary versions of some of the presented results we published in conference papers on MFCS-2015 \cite{mfcs2015} (Theorems~3-4) and MFCS-2017 \cite{mfcs2017}  (Theorems~6-7).}
}

\author{Bruno Durand and Andrei Romashchenko}

\maketitle

\begin{abstract}
We study multidimensional \emph{minimal} and \emph{quasiperiodic} shifts of finite type. We prove for these classes several results that were
previously known for the shifts of finite type in general, without restriction.
 We show that some quasiperiodic shifts of finite type admit only non-computable configurations; we characterize the classes of Turing degrees that can be represented by quasiperiodic shifts of finite type.
 We also transpose to the classes of minimal/quasiperiodic shifts of finite type some  results on subdynamics previously known for the effective shifts without restrictions: every \emph{effective} minimal (quasiperiodic) shift of dimension $d$ can be represented as a projection of a subdynamics of a  minimal (respectively, quasiperiodic) shift \emph{of finite type} of dimension $d+1$. 
\end{abstract}


\section{Introduction}

In this paper we study the multi-dimensional \emph{shifts}, 
i.e.,  the shift-invariant and topologically closed sets of configurations in  $\mathbb{Z}^d$ over a finite alphabet. 
The \emph{minimal shifts} are those shifts in which all configurations contain exactly  the same finite patterns;
the \emph{quasiperiodic shifts} are the shifts where each configuration contains all of its  finite patterns  infinitely often,  at least once in 
every large enough region. 

Two  classes of  shifts play a prominent role in symbolic dynamics, in language theory,  and in the theory of computability:
the \emph{shifts of finite type} (obtained by forbidding a finite number of finite patterns) and the \emph{effective shifts}
(obtained by forbidding a computable set of finite patterns). 

The class of shifts of finite type is known to be very rich --- even a finite set of simple local rules can induce a rather sophisticated global structure. It is known that some (non-empty) multi-dimensional shifts of finite type admit only aperiodic (see \cite{berger}) or even only non-computable (\cite{nonrecursive1,nonrecursive2}) configurations.
The value of the  combinatorial entropy of a shift of finite type can be any right computable real number $h\ge 0$, see \cite{hochman-meyerovitch}. 
Recent results (see \cite{hochman, drs, aubrun-sablik}) show that if we look at the projective subdynamics of a shift of finite type of some dimension $d>1$, we can obtain any effective shift of a dimension
below $d$.

The proofs of the results mentioned above are based on embedding some computation in the structure of a shift of finite type.  These embedding use rather tricky combinatorial gadgets that permit marrying the `computational' techniques with the framework of symbolic dynamics.
Though the implied computational tricks may look natural to a computer scientist, they give in the end constructions that are rather complex and artificial   if we look at them as objects of dynamical systems theory. 

Thus, a natural question is whether the phenomena mentioned above hold for `simpler' and `more natural' shifts of finite type. A formal version of this problem is  to transpose some results known for shifts of finite type \emph{in general} to some sort of \emph{primary}  shifts. More technically, in this paper we deal with \emph{minimal}, \emph{quasiperiodic}, and \emph{transitive} shifts. Our results can be subdivided into three groups:
\begin{itemize}
\item We prove that some quasiperiodic shifts of finite type admit only non-computable configurations. Moreover, we characterize the classes of Turing degrees that may correspond to a quasiperiodic shift of finite type.
\item We transpose the  characterization  by Hochman and Meyerovitch of the entropies of multidimensional shifts of finite type to the class of transitive shifts. (The same result was recently proven by Gangloff and Sablik in \cite{gangloff-sablik} with a different technique.)
\item We extend to the classes of minimal/quasiperiodic shifts of finite type some known results on subdynamics:
every \emph{effective} minimal (quasiperiodic) shift of dimension $d$ can be represented as a projection of a subdynamics of a  minimal (respectively, quasiperiodic) shift \emph{of finite type} of dimension $d+1$, which answers positively a question  by E.~Jeandel, \cite{jeandel-email}
\end{itemize}
All constructions in this paper involve the technique of self-simulating tilings developed in~\cite{drs} (see also variants of this technique in \cite{zinoviadis,westrick}).

\subsection{Notation and basic definitions}

\paragraph{Shifts.}
Let $\Sigma$ be a finite set (an alphabet). Fix an integer $d>0$. A $\Sigma$-\emph{configuration}  (or just a \emph{configuration} if $\Sigma$ is clear from the context) on $\mathbb{Z}^d$ is a mapping
$
 \mathbf{f} \ :\ \mathbb{Z}^d \to \Sigma,
$
i.e.,  a coloring of $\mathbb{Z}^d$ by ``colors'' from $\Sigma$. 
A $\mathbb{Z}^d$-\emph{shift}  (or just a \emph{shift}) is a set of configurations that is 
 (i)~translation invariant (with respect to the translations along each  coordinate axis), and
 (ii)~closed in Cantor's topology. 
 
 If a $\mathbb{Z}^d$-shift ${\cal S}_1$ is subset of a $\mathbb{Z}^d$-shift ${\cal S}_2$, we say that ${\cal S}_1$ is a \emph{subshift} of ${\cal S}_2$
 The entire space $\Sigma^{\mathbb{Z}^d}$ is itself a shift (where no pattern is forbidden) and is called \emph{the full shift}. Thus, every 
 $\mathbb{Z}^d$-shift is a subshift of the full shift.

 A \emph{pattern} is a mapping from a finite subset of $\mathbb{Z}^d$ to $\Sigma$ (a coloring of a finite set of $\mathbb{Z}^d$); this set is called the support of the pattern. 
 We say that a pattern $P$ \emph{appears} in a configuration  $ \mathbf{f}(\bar x)$
if for some  $\bar c\in \mathbb{Z}^d$
the pattern $P$  coincides with the restriction of the shifted configuration $ \mathbf{f}_{\bar c}(\bar x):= \mathbf{f}(\bar x+\bar c)$ to the support of this pattern.
A pattern that appears in some configuration of a shift is called \emph{globally admissible}. 

 Every shift is determined by the corresponding set of forbidden finite patterns $\cal F$ (a configuration belongs to the shift if and only if no patterns from $\cal F$ appear in this configuration). 

A shift is called  \emph{effective} (or \emph{effectively closed}) if it can be defined by a computably enumerable  set of forbidden patterns. A shift is called a \emph{shift of finite type} (SFT) if it can be defined by a finite set of forbidden patterns.   

 \paragraph{Wang tilings.}
A special class of two-dimensional SFT is defined in terms of \emph{Wang tiles}. In this case, we interpret the alphabet $\Sigma$ as a set of \emph{tiles}, i.e.,  a set of unit squares with colored sides, assuming that all colors belong to some finite set $C$ (we assign one color to each side of a tile, so technically $\Sigma$ is a subset of $C^4$). A (valid) \emph{tiling} is  a set of all configurations
 $
 \mathbf{f}\ :\ \mathbb{Z}^2\to \Sigma
 $
where every two neighboring tiles match, i.e., share the same color on adjacent sides. 
Wang tiles are powerful enough to simulate any SFT in a  very strong sense: for each SFT $\cal S$ there exists a set of Wang tiles $\tau$ such that the 
set of all $\tau$-tilings is isomorphic to $\cal S$.  In  this paper we mainly use the formalism of tilings since Wang tiles are better adapted for explaining our techniques of self-simulation.

 \paragraph{Dynamics and subdynamics.}
Every shift ${\cal S}\subset \Sigma^{\mathbb{Z}^d}$ can be interpreted as a dynamical system. Indeed, there are $d$  translations along  the coordinate axes, and each of these translations  maps  $\cal S$ to itself. Therefore, the group $\mathbb{Z}^d$ naturally acts on $\cal S$. 

Let $\cal S$  be a shift on $\mathbb{Z}^d$ and $L$ be  $k$-dimensional sub-lattice  in $\mathbb{Z}^d$ (i.e., $L$ must be an additive  subgroup of $\mathbb{Z}^d$ that is isomorphic to $\mathbb{Z}^k$). Then
the $L$-\emph{projective subdynamics}  ${\cal S}_L$ of $\cal S$ is the set of configurations of $\cal S$ restricted on $L$.
The $L$-projective subdynamics of a $\mathbb{Z}^d$-shift can be understood as a $\mathbb{Z}^k$-shift
(note that $L$ naturally acts on ${\cal S}_L$).
 In particular, for every $d'<d$ we have a $\mathbb{Z}^{d'}$-projective subdynamics on the shift $\cal S$, generated by the lattice on the first $d'$ coordinate axis.

A configuration $\mathbf{x}$ is called \emph{recurrent} if every pattern that appears in $\mathbf{x}$ at least once, must then appear in this configuration
infinitely often.  

A shift $\cal S$ is called \emph{transitive} if there exists a configuration $x\in{\cal S}$ that contains every finite pattern
that appears in  at least one configuration $y\in{\cal S}$.

 \paragraph{Quasiperiodicity and minimality.}
A configuration $\mathbf{x}$ is called \emph{quasiperiodic} (or \emph{uniformly recurrent}) if every pattern $P$ that appears in $\mathbf{x}$ at least once must appear in every large enough cube $Q$  in $\mathbf{x}$. Note that every periodic configuration is also quasiperiodic. A \emph{quasiperiodic} shift is a shift that contains only quasiperiodic configurations.

Given a configuration $\mathbf{x}$, a \emph{function of a quasiperiodicity} for $\mathbf{x}$ is a mapping\label{def-function-of-quasiperiodicity}
 $
  \varphi \ : \ \mathbb{N} \to \mathbb{N}\cup \{\infty\}
 $  
such that  every finite pattern of size (diameter) $n$  either never appears in  $\mathbf{x}$ or appears in every cube of size  
$\varphi(n)$ in $\mathbf{x}$ (see \cite{bruno}). We assume $\varphi(n)=\infty$ if some pattern $P$ of size $n$ appears in $\mathbf{x}$ but there exist arbitrarily large areas in $\mathbf{x}$ that are free of $P$. By definition, for a quasiperiodic $\mathbf{x}$ we have  $\varphi(n)<\infty$ for all $n$.
We say that a  shift $\cal S$ has a function of quasiperiodicity if 
there exists a function $ \varphi(n)$ (finite  for all $n$) that is a function of a quasiperiodicity
for every configuration  in ${\cal S}$. 

A shift is called \emph{minimal} if it  contains no non-trivial subshifts  (except the empty set and itself). 
A shift ${\cal S}$ is minimal if and only if all its configurations contain exactly the same patterns.
If a shift is minimal, then it is quasiperiodic (the converse is not true).

Since each minimal  shift $\cal S$ is quasiperiodic,  for every minimal  shift $\cal S$ there is a function of quasiperiodicity  $\varphi(n)$
that is finite for all $n$. 
Besides,  for an effective  minimal  shift, the set of all finite patterns (that can appear in any configuration)  is computable, see \cite{hochman, ballier-jeandel-2008}. 
From this fact it follows that every effective and minimal  shift  contains some computable  configuration.
Indeed, with an algorithm that checks whether a given pattern appears in every $\mathbf{x}\in{\cal S}$, we can  incrementally (and algorithmically) increase a finite pattern, maintaining the property  that this pattern appears in every configuration in $\cal S$.

\paragraph{Topological entropy.}
Let $Q_k$ be the $d$-dimensional \emph{cube} of size $k$,  $$Q_k:=\{0,1,\ldots,k-1\}^d.$$
For a shift  $\cal S$ in $\Sigma^{\mathbb{Z}^d}$, we denote by $N_{\cal S}(k)$ the number of distinct $\Sigma$-colorings of $Q_k$
that appear as a pattern  in configurations from $\cal S$.
The topological entropy of $\cal S$ is defined by
$$
h({\cal S})  = \lim \limits_{k\to \infty}\frac{ \log N_{\cal S}(k)}{|Q_k|}
$$
with the logarithm to base $2$.
\textup(The limit above exists for any shift.\textup)

\subsection{The main results}

The main results of this paper can be subdivided into three groups:  constructions combining quasiperiodicity with high computational complexity, a construction of a transitive SFT with a given topological entropy, and subdynamics of minimal and quasiperiodic SFT. In what follows we state our results more precisely.

\subsubsection{Quasiperiodicity is compatible with non-computability}

A configuration 
$
 \mathbf{f} \ :\ \mathbb{Z}^d \to \Sigma,
$
is called \emph{periodic} if there exists a non-zero $c\in \mathbb{Z}^d$ such that 
$
 \mathbf{f}(x+c) =  \mathbf{f}(x)
$
for all $x$. Otherwise a configuration is called \emph{aperiodic}.
In the classic paper \cite{berger}, Berger came up with a shift of finite type where all configurations are aperiodic. 
\begin{theorem}[Berger]
\label{thm-berger}
For every $d>1$ there exists a non-empty  SFT on $\mathbb{Z}^d$ where   each configuration is aperiodic.
\end{theorem}
This construction, as well as a simpler one proposed by Robinson in \cite{robinson}, is quite sophisticated, with tricky combinatorial gadgets embedded in the structure of all configurations of the shift. 
So the SFT obtained in these seminal papers (and many subsequent constructions elaborating the ideas of Berger and Robinson)
look rather messy in terms of  dynamical systems. The SFTs proposed by  Berger and Robinson
are   not ``simple'' as a topological object, neither minimal nor transitive\footnote{%
The basic construction of an aperiodic SFT by Robinson admits configurations that consist of two  half-planes 
which can be arbitrarily  shifted one with respect to the other. 
Robinson referred to  this phenomenon as a \emph{fault} in a tiling, see  \cite[Section~3]{robinson}.
This makes the SFT non-minimal and even non-transitive.

The shifts obtained by enriching  the generic  construction of Robinson (tilings with embedded computations) may be even more irregular.
If they admit configurations with potentially different computations (which is  unavoidable in some applications),
then it is difficult  to enforce the appearance  in \emph{one single  configuration} of all fragments of computations that 
potentially might appear in \emph{at least one} valid configuration. Recently
several authors proposed nontrivial regularizations of  Robinson's construction overcoming the aforementioned obstacles, see below.}.

Later it was shown that the property of aperiodicity can be enforced for \emph{minimal} shifts of finite type. In other words, the combinatorial complexity (the property of aperiodicity of an SFT) can be combined with topological simplicity (the property of minimality), as in the following theorem. 
\begin{theorem}[Ballier--Ollinger]
\label{thm-alexis-nicolas}
For every $d>1$ there exists a non-empty  SFT on $\mathbb{Z}^d$ where   each configuration is aperiodic and quasiperiodic \textup(and even minimal\textup).
\end{theorem}
(This result was proven in \cite{alexis} for a tile set $\tau$ constructed in \cite{ollinger}.
Other examples of minimal aperiodic SFTs were suggested in in \cite{goodman-strauss, labbe}.)
Our next result  in some sense strengthens Theorem~\ref{thm-alexis-nicolas}: we claim that there exists a  quasiperiodic SFT that contains only non-computable configurations.
\begin{theorem}\label{thm1}
For every $d>1$ there exists a non-empty  SFT on $\mathbb{Z}^d$  where all configurations
are  non-computable and quasiperiodic.
\end{theorem}
Note that we cannot strengthen Theorem~\ref{thm1} further and change \emph{quasiperiodicity} to \emph{minimality}:
in every minimal SFT the set of globally admissible patterns is decidable, and therefore there exists a computable configuration, 
see \cite{ballier-jeandel-2008} and \cite{hochman}.

With a similar technique we prove the following (more general) result, which characterize the classes of Turing degrees that can be represented by quasiperiodic SFTs.
\begin{theorem}\label{thm1-bis}
For every effectively closed set  $\cal A$ and for every $d>1$ there exists a non-empty  SFT $\cal S$  in $\mathbb{Z}^d$ where all configurations 
are quasiperiodic,
and the Turing degrees of all configurations in $\cal S$ form exactly the upper closure of  ${\cal A}$ 
\textup(defined as  the set of all Turing degrees $d$ such that $d\ge_T \mathbf{x} $ for at least one $\mathbf{x}\in {\cal A}$\textup).
\end{theorem}
\begin{rmrk}
The Turing degree spectrum
of a non effective minimal shift can be much more complex than what we get in Theorem~\ref{thm1-bis}, 
see \cite{hochman-vanier}. 
\end{rmrk}

\subsubsection{The  Hochman--Meyerovitch theorem on the entropy of SFTs}

Hochman and Meyerovitch showed in \cite{hochman-meyerovitch} that a real number $h\ge 0$  is the topological entropy of some SFT
 if and only if $h$ is right recursively enumerable. They raised the question whether the same property holds for the class of \emph{transitive} shifts.
 Recently Gangloff and Sablik \cite{gangloff-sablik} answered  this question positively using a construction based on Robinson's aperiodic tilings. 
 We prove the same result (Theorem~\ref{t:transitiv-hochman-meyerovich} below) using a technique similar to the proof of Theorem~\ref{thm1}.

\begin{theorem}\label{t:transitiv-hochman-meyerovich}
For every integer $d>1$ and 
for every  nonnegative right recursively enumerable real $h\ge 0$ there exists a  \emph{transitive} SFT on $\mathbb{Z}^d$ with the topological entropy $h$.
\end{theorem}

\subsubsection{Subdynamics of minimal and quasiperiodic shifts}

Our next theorem claims that the subdynamics of an effective quasiperiodic SFT can be very rich. More specifically, we prove that
 every effective quasiperiodic $\mathbb{Z}^d$-shift can be simulated by a quasiperiodic SFT on $\mathbb{Z}^{d+1}$.
By \emph{simulation} we mean a factor of  the subdynamics of a shift on $\mathbb{Z}^{d+1}$ (where the subdynamics can be understood as the restriction of a shift on the first $d$ coordinate axis). We proceed with a  formal definition:

\begin{definition}\label{def:subdynamics}
We say that a shift $\cal A$ on $\mathbb{Z}^{d}$ is \emph{simulated} by 
a shift $\cal B$ on $\mathbb{Z}^{d+1}$ if there exists a projection $\pi \ : \ \Sigma_B\to \Sigma_A$ such that for every configuration
 $$
  \textbf{f}\ : \ \mathbb{Z}^{d+1} \to \Sigma_B
 $$
from $\cal B$ and for all $i_1,\ldots,i_d,j,j'$  
we have $\pi (\textbf{f}(i_1,\ldots,i_d,j)) = \pi(\textbf{f}(i_1,\ldots,i_d,j'))$ 
\textup(i.e., the projection $\pi$ takes a constant value along each column $(i_1,\ldots,i_d,*)$, see Fig.~\ref{fig-tiling-with-embedded-letters}\textup), 
and the resulting $d$-dimensional configuration 
$$
 \{ \pi(\textbf{f}(i_1,\ldots,i_d,*)) \}
$$
belongs to $\cal A$; moreover, each configuration of $\cal A$ can be represented in this way by some configuration of $\cal B$.  
Informally, we can say that each configuration from $\cal B$ encodes a configuration from $\cal A$, and each configuration from $\cal A$ is encoded by some configuration from $\cal B$. 
\end{definition}
\begin{figure}
\begin{center}
\includegraphics[scale=0.8]{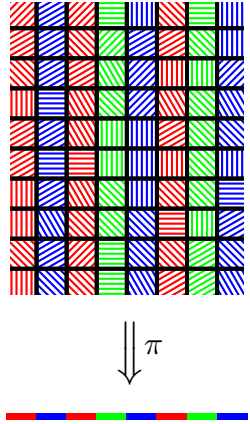}
\caption{In this example, each cell of a two-dimensional configuration has two characteristics: the color and the direction of hatching. The color is maintained unchanged along each vertical line. The projection $\pi$ maps each column to its color.}\label{fig-tiling-with-embedded-letters}
\end{center}
\end{figure}
\begin{theorem}\label{thm-main}
(a)~Let $\cal A$ be an effective quasiperiodic $\mathbb{Z}^d$-shift over some alphabet $\Sigma_A$. Then there exists a quasiperiodic SFT $\cal B$ 
 \textup(over another alphabet $\Sigma_B$\textup) of dimension $d+1$ such that
$\cal A$ is \emph{simulated} by $\cal B$ in the sense of Definition~\ref{def:subdynamics}. 

(b)~Let $\cal A$ be a $\mathbb{Z}^d$-shift simulated in the sense of Definition~\ref{def:subdynamics} 
by some quasiperiodic SFT $\cal B$ of dimension $d+1$. Then $\cal A$ is  effective and quasiperiodic. 
\end{theorem}
\begin{rmrk}
The parts (a) and (b) of Theorem~\ref{thm-main} are the \emph{if} and \emph{only if} parts of the following characterization:
an effective shift is quasiperiodic if and only if it is simulated by a quasiperiodic SFT of dimension higher by $1$.
\end{rmrk}

\smallskip

\noindent
A similar result holds for effective minimal shifts:
\begin{theorem}
\label{thm-main-min}
(a)~For every effective minimal $\mathbb{Z}^d$-shift $\cal A$ there exists a minimal SFT $\cal B$ 
 in $\mathbb{Z}^{d+1}$ such that
$\cal A$ is simulated by $\cal B$ in the sense of Definition~\ref{def:subdynamics}. 

(b)~Let $\cal A$ be a $\mathbb{Z}^d$-shift simulated in the sense of Definition~\ref{def:subdynamics} 
by some minimal SFT $\cal B$ of dimension $d+1$. Then $\cal A$ is effective and minimal. 
\end{theorem}
\begin{rmrk}
The \emph{only if} part of this theorem (Theorem~\ref{thm-main-min}(b)) is due to
the fact that the notion of \emph{simulation} in Definition~\ref{def:subdynamics} is rather restrictive. 
We should keep in mind that in general a $d$-dimensional subdynamics of a minimal $\mathbb{Z}^{d+1}$-shift is not necessarily minimal.
\end{rmrk}

Theorem~\ref{thm-main} implies the following rather surprising corollary: a quasiperiodic $\mathbb{Z}^2$-SFT can have a highly ``complex'' language of patterns:
\begin{corollary}\label{thm-kolmogorov}
There exists a quasiperiodic SFT $\cal A$ of dimension $2$ such that
the Kolmogorov complexity of every $(N\times N)$-pattern in every configuration of $\cal A$ is $\Omega(N)$.
\end{corollary}
In other words, a quasiperiodic $\mathbb{Z}^2$-SFT can have extremely ``complex'' languages of patterns.

\begin{rmrk}A standalone pattern of size $N\times N$ over an alphabet $\Sigma$ (with at least two letters) can have Kolmogorov
complexity up to $\Theta(N^2)$. However, this density of information cannot be enforced by local rules, because in every SFT on $\mathbb{Z}^2$ there exists
a configuration such that the Kolmogorov complexity of all $N\times N$-patterns is bounded by $O(N)$, see \cite{dls}. Thus, the lower bound $\Omega(N)$
 in Corollary~\ref{thm-kolmogorov} is optimal in the class of all SFT on $\mathbb{Z}^2$.
\end{rmrk}

\begin{rmrk}
Every effective (effectively closed) \emph{minimal} shift is computable: given a pattern, we can algorithmically decide whether it belongs to the configurations of the shift, \cite{hochman}, 
which is in general not the case for effective quasiperiodic shift. 
On the other hand, it is known that patterns of high Kolmogorov complexity cannot be found algorithmically. Thus Corollary~\ref{thm-kolmogorov} cannot be extended to the class of minimal SFT.
\end{rmrk}

\subsection{General remarks and organization of this paper}

In the theorems stated above, we claim something about (quasiperiodic, minimal, transitive) SFTs. In the proofs we deal mostly with tilings, which are a very special type of SFT. Since the principal results of the paper are positive statements (we claim that SFTs with some specific properties do exist), the focus on tilings does not restrict the generality. On the other hand, the formalism of Wang tiles perfectly matches  the constructions of self-similar and self-simulating shifts of finite type, which are the main technique of this paper. To simplify the notation and make the argument more visual, in what follows we focus on the case $d=2$. The proofs extend to any $d>1$ in a straightforward way, \emph{mutatis mutandis}.

The central idea of our arguments is the notion of \emph{self-simulation}, which goes back to 
\cite{gacs-self-similar-automata} and which was extensively developed in the context of symbolic dynamics in \cite{drs}. The technique of hierarchical self-simulating tilings is quite generic and flexible. The drawback of this approach is that it is hard to isolate the core technique from features specific to some particular application. In every specific result we cannot just cite the \emph{statement} of a previously known theorem about self-simulating tilings: rather, we need to reemploy the \emph{constructions} from the proofs of a previously known theorem (embedding some new gadgets in the previously known general scheme). 
This makes the proofs long and somewhat cumbersome. To help the reader, we have tried to make this paper self-contained and explain here the entire machinery of hierarchical self-simulating tilings. 
The exposition of the main results becomes itself ``hierarchical'' and ``self-similar.'' We start with a general perspective of our technique (illustrating it by proofs of several classic results). Then in each succeeding section, we show how to adjust and extend this general construction to obtain this or that new result.
We encourage the reader familiar with the concept of self-similar tilings to skip Section~2, which may look oversimplified, and go directly to Sections~3--7.

\section{The general framework of self-simulating SFT}

In this section we recall the principal ingredients of the technique of \emph{self-simulating tile sets}. We start this section with a very basic version of a construction of self-simulating tile sets from \cite{drs}. This part of the construction will be enough, in particular, to obtain a proof of the classic Theorem~\ref{thm-berger}. Later, we will extend this construction and adapt it to prove much stronger statements.

\subsection{The relation of simulation for tile sets}\label{ss:2-1}

Let $\tau$ be a tile set and $N>1$ be an integer. We call a \emph{macro-tile} an $N \times N $ square tiled by matching tiles from $\tau$. Every side of a $\tau$-macro-tile contains a sequence of $N$ colors (of tiles from $\tau$); we refer to this sequence as a \emph{macro-color}. \label{def-macro-color}
Further, let $T$ be some set of $\tau$-macro-tiles (of size $N\times N$). We say that $\tau$ \emph{implements} $T$ with a \emph{zoom factor} $N$
if 
\begin{itemize}
\item some $\tau$-tilings exist, and 
\item for every $\tau$-tiling there exists a unique lattice of vertical and horizontal lines that cuts this tiling into $N\times N$ macro-tiles from $T$. 
\end{itemize}
A tile set $\tau$ \emph{simulates} another tile set $\rho$ if $\tau$ implements a set of macro-tiles $T$ (with a zoom factor $N>1$) that is isomorphic to $\rho$, i.e., there exists a one-to-one correspondence between $\rho$ and $T$ such that the matching pairs of $\rho$-tiles correspond exactly to the matching pairs of $T$-macro-tiles.  A tile set $\tau$ is called \emph{self-similar} if it simulates itself.

If a tile set $\tau$ is self-similar, then all $\tau$-tilings have a hierarchical structure. Indeed, each $\tau$-tiling can be uniquely split into $N\times N$ macro-tiles from a set $T$, and these macro-tiles are isomorphic to the initial tile set $\tau$.  Further, the grid of macro-tiles can be uniquely grouped into blocks of size $N^2\times N^2$, where each block is a macro-tile of rank $2$ (again, the set of all macro-tiles of rank $2$ is isomorphic to the initial tile set $\tau$), etc. It is not hard to deduce  that a self-similar tile set $\tau$ has only aperiodic tilings (for more details, see  \cite{drs}).  Below, we discuss a generic construction of self-similar tile sets.

\subsection{Simulating a tile set defined by a Turing machine}
\label{section:2.2}
 
Let us have a tile set $\rho$. In what follows, we present a general construction that allows to simulate  $\rho$ by some other tile set $\tau$, with a large enough zoom factor $N$. The number of tiles in the simulating tile set $\tau$ will be $O(N^2)$, and the constant in the $O(\cdot)$-notation does not depend on the simulated $\rho$.

We assume that  each color is a string of $k$ bits (i.e., the set of colors $C \subset \{0,1\}^k$) and the set of tiles $\rho \subset C^4$ is presented by a predicate $P(c_1,c_2,c_3,c_4)$ (the predicate is true if and only if the quadruple $(c_1,c_2,c_3,c_4)$ corresponds to a tile from $\rho$). Suppose we have a Turing machine $\cal M$ that computes $P$. (It might look wasteful to construct a Turing machine that computes a predicate with a finite domain, but we will see that this kind of abstraction is useful.)
Now we  construct in parallel a tile set $\tau$ and a set of $\tau$-macro-tiles that simulate the given~$\rho$.

When constructing a tile set $\tau$, we will keep in mind the desired structure  of $\tau$-macro-tiles (that should simulate the given tile set $\rho$).
We require that  each tile in $\tau$ ``knows'' its coordinates modulo $N$ in the tiling. This information is included in the tile's colors. More precisely, for a tile that is supposed to have  coordinates $(i,j)$ modulo $N$, the colors on the left and on the bottom sides should involve $(i,j)$, the color on the right side should involve $(i+1\mod N, j)$, and the color on the top side involve $(i, j+1 \mod N)$, see Fig.~\ref{fig-0}.
\begin{figure}
\begin{center}
\includegraphics[scale=0.7]{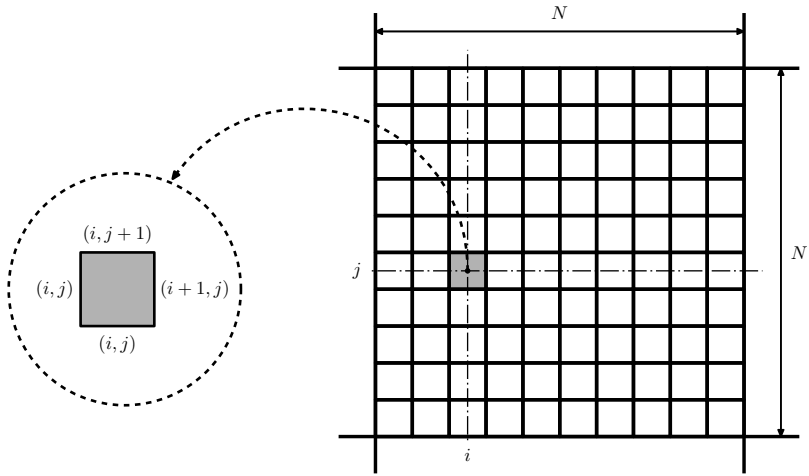}
\caption{The basic structure of a macro-tile: A block of size $N\times N$ consists of tiles,  the coordinates of which in this block 
(a pair of integers between $0$ and $N-1$) are written as ``colors'' of the left and the bottom side of each tile. On the right and on the bottom sides
of every tile, one of the coordinates is incremented (modulo $N$).}\label{fig-0}
\end{center}
 \end{figure}

This means that every $\tau$-tiling can be uniquely split into blocks (macro-tiles) of size $N\times N$, where the coordinates of the cells range from $(0,0)$ in the bottom-left corner to $(N-1,N-1)$ in top-right corner, as shown in Fig.~\ref{fig-0}. Intuitively, each tile ``knows'' its position in the corresponding macro-tile.
We will require that in addition to the coordinates, each tile in $\tau$  has some supplementary information   encoded in the colors on its sides.  
On the border of a macro-tile (where one of the coordinates of a cell is zero), we assign  to the colors of each cell one additional bit of information.
Thus, for each macro-tile of size $N\times N$ the corresponding macro-colors (in the sense of the definition in Section~\ref{def-macro-color},  p.~\pageref{def-macro-color}) can be  represented as strings of $N$ zeros and ones.  

The number of bits encoding a macro-color (an $N$-bit string representing a macro-color of a macro-tile of size $N\times N$) is excessively large for our future constructions. We choose an integer number $k$  ($k\ll N$) and allocate   in the middle of a macro-tile's sides $k$ positions;  we make them represent colors from $C$. The other $(N-k)$ bits on the sides of a macro-tile are ``dummy'' (formally speaking, we set them all to zero),
see Fig.~\ref{fig-macrocolors}.

\begin{figure}
\begin{center}
\includegraphics[scale=0.7]{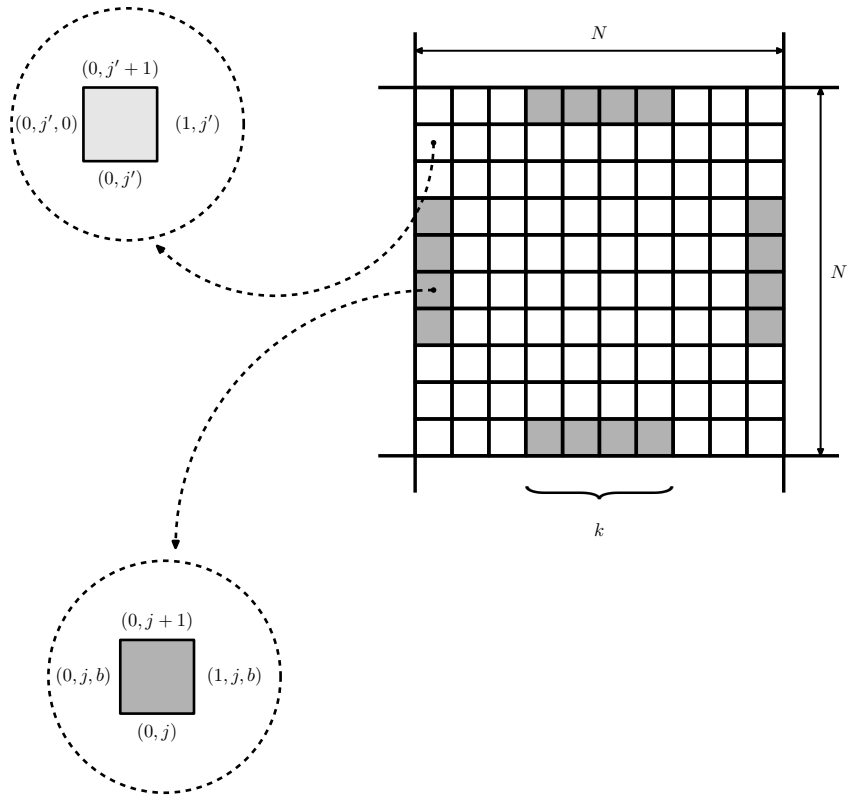}
\caption{We require that  tiles on the margins of an $N\times N$ macro-tile  carry one supplementary bit.
For example, the $j$-th tile on the left margin of a macro-tile contains in the color of its left side a triple $(0,j,b)$, 
where $(0,j)$ represent the coordinate of this tile in the macro-tile,  and the bit $b$ may be equal to $0$ or $1$.
Thus, each macro-color can be now represented by a sequence of $N$ bits embedded in the sides of the tiles on the
corresponding margin of a macro-tile. \newline %
\protect\rule{2em}{0ex} 
We assume that these supplementary bits are nontrivial only 
for the $k$ tiles in the middle of each macro-tile's margin. In the figure, these tiles are shown in gray, e.g., the $j$-th tile on the left
side of a macro-tile carries a bit $b$, which may be equal to $0$ or $1$. The other  tiles on macro-tile's margin  (e.g., the $j'$-th
tile on the left side of a macro-tile) do not contribute in the macro-color  --- their supplementary bits are always set to $0$. 
Thus, each macro-color is actually determined by a sequence of only $k$ bits.
\newline %
\protect\rule{2em}{0ex} 
The bits representing the macro-colors are ``propagated'' in the macro-tile. In what follows we discuss how this information is ``processed'' inside.
\newline %
\protect\rule{2em}{0ex} 
Note that the value of $k$ in our construction is usually \emph{much} less than $N$, e.g., $k=O(\log N)$.}
\label{fig-macrocolors}
\end{center}
 \end{figure}
  
 Now we introduce additional restrictions on the tiles in $\tau$ that will guarantee that the macro-colors on the macro-tiles satisfy the ``simulated'' relation $P$. To this end, we ensure that bits from the macro-tile side are transferred to the central part of the tile, and the central part of a macro-tile is used to simulate  a  computation of the predicate $P$.
We fix which cells in a macro-tile are ``communication wires'' 
and then require that these tiles carry the same (transferred) bit on two sides, see Fig.~\ref{fig-1-wires}. 
\begin{figure}
\begin{center}
\includegraphics[scale=0.30]{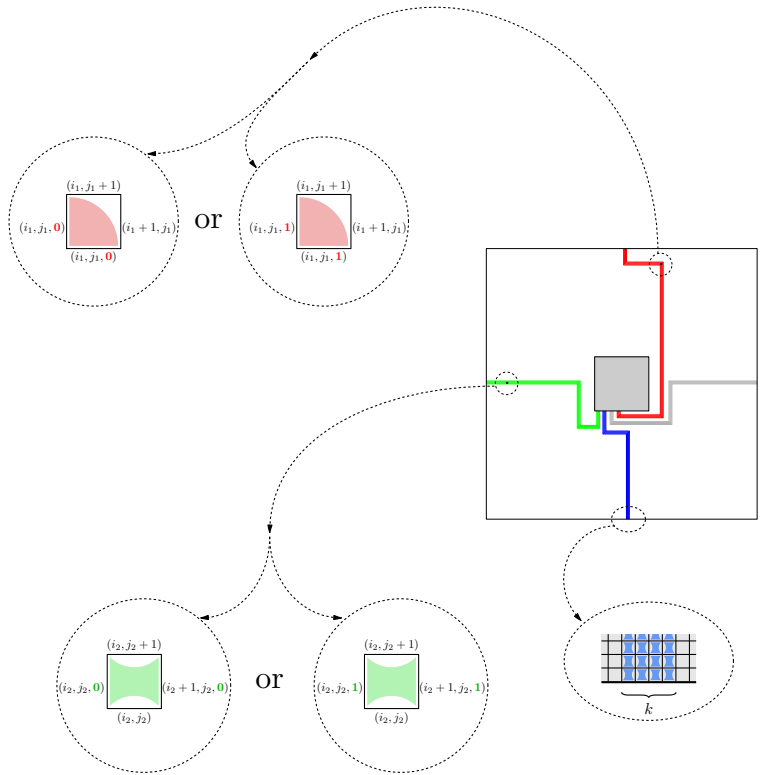}
\caption{In the middle of each margin of a macro-tile we allocate $k$ tiles that contain the bits representing a macro-color. These tiles are plugged
into ``communication wires'' that transfer these data inside of a macro-tile.
The value of $k$ (which is the width of the red, green, blue, and gray stripes in the figure) is much less than the size of the macro-tile.
\newline %
\protect\rule{2em}{0ex} 
In the figure,  we zoom in  two tiles involved in the ``communication wires.''
The tile with coordinates $(i_1,j_1)$ is a part of one of the communication wires   shown in red, which transfer the bits of the top macro-color.
As a part of a communication cable, this tile  transfers a bit value \emph{zero} or \emph{one}. This tile is placed on a corner  of a wire,  so  it ``conducts'' a bit value  from its left side  to its bottom side (the bit values  embedded in its left and bottom sides must be both zeros or both ones).
Similarly, the tile with coordinates $(i_2,j_2)$ is a part of a  communication wires shown in green, which transfer the bits of the left macro-color. This tile is involved in a horizontal part of a wire, and it ``conducts'' a bit value  \emph{zero} or \emph{one} from the left to the right
(so the bit values embedded in its left and right sides must be equal to each other). 
\newline %
\protect\rule{2em}{0ex} 
The coordinates and the values of the transferred bits are embedded in the ``colors'' on the sides of the tiles. 
The colors inside individual tiles (red, green, blue, gray in this figure) are only illustrative and show the functional role of each tile in a macro-tile.
}\label{fig-1-wires}
\end{center}
 \end{figure}

The central part of a macro-tile (of size, say, $m\times m$, where $m\ll N$) should represent a space-time diagram of the machine $\cal M$ (the tape is horizontal, and time goes up), see Fig.~\ref{fig-macrotile-tm}.
 \begin{figure}
\begin{center}
\includegraphics[scale=0.3]{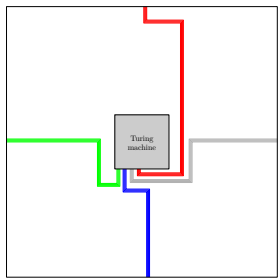}
\caption{A macro-tile with a space-time diagram of a Turing machine in the middle part.}\label{fig-macrotile-tm}
\end{center}
 \end{figure}
This Turing machine processes the quadruple of inputs, which are  the $k$-bit strings representing the macro-colors of this macro-tile. 

Let us explain in more detail how we represent the computation of a Turing machine in a tiling.
First of all, we assume that the machine has a single tape. We  understand the \emph{space-time diagram} of a Turing machine in a pretty standard way,
as a table where each vertical column corresponds to one cell on the tape of the machine, and  
each horizontal row of the diagram represents an instance of a configuration of the Turing machine. 
For each row of the diagram, we 
\begin{itemize}
\item  specify  for each cell (within a bounded part of the tape) the letter written in this position on the tape, 
\item specify with a special mark the position of the read head, and 
\item write the index of the internal state of the machine into the cell where the read head is currently located.  
\end{itemize}
Each next row of the diagram represents the  configuration of the machine at the next step of computation (once again, we assume that time goes up).
Thus, the entire diagram is  determined by its bottom line (with the input data of the machine).

The property of being a valid space-time diagram is defined locally,  so we can easily represent such diagram of a given Turing machine by local matching rules for tiles. The details of this representation are not very important in the sequel; we may take, for example, the representation described in \cite{allauzen-durand}. In what follows, we need only keep in mind some natural properties of the chosen representation
(which trivially hold for the  representations from \cite{allauzen-durand}):
\begin{itemize}
\item a correct tiling of a frame $m\times m$ represents a space-time diagram of  the same\footnote{In fact, it is enough to assume that a tiling of size $m\times m$ represents a space-time diagram of size  $\Omega(m) \times \Omega(m)$.
}
size $m\times m $,
\item a  correct tiling of a frame $m\times m$ with some specific bottom line can be formed  if and only if the computation (with the corresponding input data) terminates in an accepting state  in  at most $m$ steps and during this computation the read head never leaves the available finite part of the tape,
\item every $(k\times k)$-fragment\footnote{In what follows we use only a restricted version of this property. We need to be able to reconstruct every $(2\times2)$-block of tiles
given the $12$ tiles around this block, see Fig.~\ref{fig-slots-detailed} below.}   
of a correct tiling can be reconstructed by  its  borderline 
(this is the only property where we  need the Turing machine to be deterministic).\label{remark-on-determ-turing-machine-diagram}
\end{itemize}
The communication of the ``computation zone'' (the area representing a space-time diagram of a Turing machine) with the ``outside world'' is restricted to the bottom line (the input data of the computation), which must cohere with the bits representing the four macro-colors of the macro-tile. 
 
  To make all of this  construction work, the size of a macro-tile (the integer $N$) should be large enough: first, we need enough room to place the ``communication wires'' that transfer the bits of macro-colors to the ``computation zone''; second, we need enough time  and space  in the computation zone of size $m\times m$ so that all accepting computations of $\cal M$ terminate in time $m$ and on space $m$.
 
In this construction, the number of additional  bits encoded in the colors of the tiles
depends on the choice of machine $\cal M$. To avoid this dependency, we replace $\cal M$ by a fixed universal Turing machine $\cal U$ that runs a program simulating $\cal M$. Moreover, we prefer to separate the general program of a Turing machine (that involves a description of the predicate $P$ corresponding to the simulated tile set $\rho$) from the zoom factor $N$.

Technically, we assume that the tape of the universal Turing machine\footnote{%
We assume that the reader is familiar with the basic notions of computability theory. For a detailed discussion of the notion of a 
universal computable function and universal Turing machines we refer the reader to the textbooks \cite[Section~II.1]{odifreddi}
and references therein, and to \cite[Section~2.1]{shen-vereshchagin-book}.%
}
 has an additional read-only layer. Each cell of this layer carries a bit that never  changes during the computation (so in the computation zone the columns carry unchanged bits). The construction of a tile set  guarantees that these bits form two \emph{read-only input fields}: (i)~the program for $\cal M$ and (ii)~the binary expansion of an integer $N$ (which is interpreted as the zoom factor). Accordingly, the computation zone of a macro-tile represents a view of an accepting computation for that program given $N$ as one of the input, see 
Fig.~\ref{fig-macrotile-utm-n}.
\begin{figure}
\centering
\includegraphics[scale=0.62]{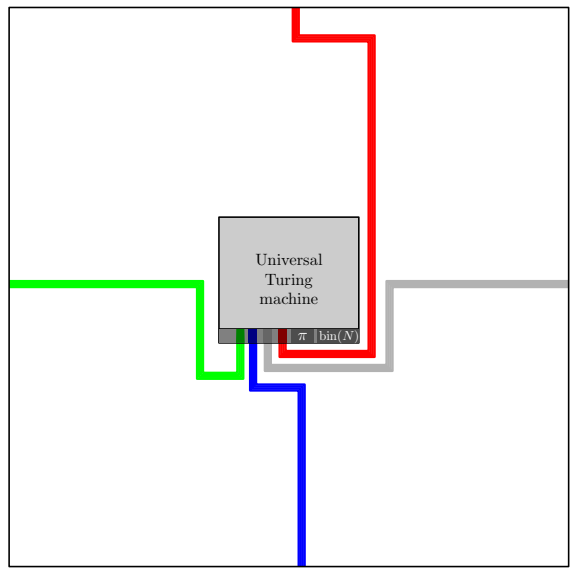}
\caption{The computation zone represents a space-time diagram of the universal Turing machine. This machine simulates program $\pi$, which gets as input the binary codes of  four macro-colors, the binary expansion  of the zoom factor $N$, and  its own text.}\label{fig-macrotile-utm-n}
\end{figure}

Thus, from now on we assume that the simulated program is given five inputs: the binary codes of four macro-colors (that are `transferred' from the sides of this macro-tile) and the binary expansion of an integer $N$ (interpreted as the zoom factor).

Without loss of generality,  we assume that the positions of the ``wires'' and the size of the ``computation zone'' in a macro-tile are chosen in some simple and natural way, and can be effectively computed given the size of the macro-tile $N$. Moreover, we may assume that the ``geometry'' of a macro-tile (the positions of the communication wires and of the computation zone) can be computed in polynomial time. That is, given the binary expansions of $N, i, j$, we can compute in time $\poly(\log N)$ the `role' played by a tile with coordinates $(i,j)$ in a macro-tile of size $N\times N$ (whether this tile is a part of a wire, or a computation zone, or none of these, and if it belongs to the computation zone, then what bits of the \emph{read-only input fields} it should carry).

In this way we obtain an \emph{explicit} construction of a tile set $\tau$ that has $O(N^2)$ tiles and simulates $\rho$. This construction works for all large enough $N$. Note that most steps of the construction of  $\tau$ do not depend on the program for $\cal M$:
The tile set $\tau$ does depend on the program simulated in the computation zone and on the choice of the zoom factor $N$.  However, this dependency is very limited. The simulated program (and, implicitly, the predicate $P$) affects only the rules for the tiles used in the bottom line of the computation zone. The colors on the sides of all other tiles are generic and do not depend on the simulated tile set~$\rho$.

The \emph{explicitness} of the described construction can be understood quite formally as follows: there exists an algorithm that takes 
the binary expansions of $N$, $k$, $m$ and a program for $\cal M$ as an input,  and returns the list of tiles in the tile set $\tau$ described above. 
The algorithm halts with an appropriate warning if $N$, $k$, or $m$ is too small for our construction. 
Moreover, for every quadruple of colors we can verify in polynomial
time (in time $\poly(\log N)$) whether this quadruple forms a tile in $\tau$ or not.

\subsection{Self-simulation with Kleene's recursion technique}

In the previous section 
we explained  how to simulate any given tile set $\rho$ by another tile set $\tau$ with a large enough zoom factor $N$. Now we want to find 
a $\rho$ such that the corresponding simulating  tile set $\tau$ is isomorphic to the simulated $\rho$. We achieve this by using an idea that comes from the proof of Kleene's recursion theorem. Roughly speaking, we employ the idea that a program can somehow \emph{access its own text} and use its bits in the computation. 
In most textbooks, the proof of Kleene's recursion theorems involves the so-called s-m-n theorem, 
which explains how an algorithm can process the source codes of other programs (given as an input), see e.g.
 \cite{odifreddi}. For a more informal discussion of the ideas behind  Kleene's recursion theorems we recommend \cite{shen-vereshchagin-book}.%

Our goal is to construct a tile set $\tau$ that simulates itself.
At the end of the previous section we observed that our construction of a simulating tile set $\tau$ has a very limited dependency on $\cal M$ 
(and, therefore, on the simulated tile set $\rho$).  
Let us fix a triple of parameters $k,m,N$; we apply the  construction from the previous section and produce the list of all tiles from the  tile set $\tau$
that simulates some tile set $\rho$ (the simulated tiles may use at most $2^k$ colors; $N$ is the zoom factor and $m$ is the size of 
the computation zone in $\tau$-macro-tiles). Since the simulated $\rho$ is not fixed yet, we cannot produce the complete list of tiles in $\tau$;
we obtain only those tiles that  \emph{do not depend on the simulated tile set},
i.e., all tiles except for the tiles that appear in the bottom line of the computation zone. In what follows we choose the missing tiles
in such a way that the resulting $\tau$ simulates itself.

To complete the construction,  we will add to $\tau$ a few other tiles  (exactly those tiles that appear in the bottom line of the computation zone). 
This construction will work for $k = 2 \log N + O(1)$ (so that we can encode $O(N^2)$ colors in strings of $k$ bits) and
$m = \poly(\log N)$ (so that we can put in the computation zone a space-time diagram of a polynomial-time computation of the Universal Turing machine).

Though the definition of $\tau$ is not finished, we already know that every valid $\tau$-tiling (if such a tiling exists) consists of an $N\times N$-grid
of macro-tiles, with $k$-bit macro-colors embedded in their sides, with communication wires and a computation zone, as shown in Fig.~\ref{fig-macrotile-utm-n}.
We want to construct a program $\pi$ that takes as an input the integers  $k$, $m$,  $N$,  and four $k$-bit strings  embodying
macro-colors, and checks that a macro-tile with the given macro-colors represents one of the  tiles in the would-be $\tau$.
Then we  embed this program in the tile set $\tau$ and complete the construction.

Now let us focus on the program $\pi$ that should perform the necessary checks 
(these checks should be simulated in the computation zone of every macro-tile).
For the macro-tiles that represent the tiles that are already included in $\tau$, the required checks  are straightforward (see the comment
at the end of the previous section). It remains to implement (and encode in the program $\pi$)  
the checks for the tiles of $\tau$ that are not defined yet. 

Let us remind that the tiles in $\tau$ that are still missing are exactly the tiles that should represent the hardwired program. 
This might seem as a paradox: 
we have to write a text of a program that handles the tiles where this program will be hardwired. On the one hand,  we need to know these tiles
to write a program; on the other hand,  we need to know the text of the program to produce the tiles.
However, we can complete the description of the program $\pi$ without knowing the missing tiles of $\tau$.
Since in our construction of a macro-tile,  the would-be program $\pi$ 
(the list of instructions interpreted by the universal Turing machine) is written on the tape of the universal machine,
 this program can be instructed to access the bits of its own ``text''  and check that if a macro-tile plays a role of a tile in the computation zone, 
 then  this macro-tile carries the correct bit of the program.
 
This is the crucial point of the construction. The algorithm implemented in $\pi$ can be explained as follows. 
The algorithm obtains as an input the bits encoding the four macro-colors of the macro-tile. 
We suppose that a macro-tile represents a tile from $\tau$, and we know that each $\tau$-tile contains a pair of coordinate $(i,j)$ modulo $N$.
We can extract these coordinates $(i,j)$ from the macro-colors of the macro-tile. If this position does \emph{not} belong to the 
bottom line of the computation zone of a macro-tile, then our task is simple:  we use the algorithm outlines at the end of the previous section.
But if the position $(i,j)$ corresponds to the bottom line in the computational zone, then the task is subtle:
in this case the $\tau$-tile represented by this macro-tile must involve some bit from the text of $\pi$
(and it should be encoded in the color on the top side of the macro-tile). 
How to check that the corresponding bit embedded in the macro-color is correct? Where do we get the ``correct'' value of this bit? 
The answer is straightforward:  the Universal Turing machine simulating $\pi$ should move its reading head to the column $j$ 
in its own computation zone and read there the required bit value. Note that there is
no chicken-and-egg paradox,  we can write the instructions for $\pi$ before we know its full text.

 Thus, we obtain the complete text (list of instructions) of the program $\pi$. This is exactly the program that should be simulated by the Universal Turing
 machine in the computational zone of each macro-tile, so it is the program whose text must be written on the bottom line of the computation zone. 
 This program provides us with the missing part of the tile set $\tau$ 
 (we supplement the tile set $\tau$ with the tiles that represent  in the bottom line of the computation zone the text of $\pi$).

It remains to choose the parameters $N$ and $m$. We need them to be large enough so that the computation described above (which deals with inputs of size $O(\log N)$) can fit in the computation zone. The computations are rather simple (polynomial in the input size, i.e., polynomial in $O(\log N))$, so they  fit in the space and time bounded by $m=\poly(\log N)$.  Thus, we set $m(N)=\poly(\log N)$ for some specific polynomial that is not too small (e.g., $m:=(\log N)^3$ is enough) and choose $N$ large enough so that $m(N)\ll N$, and the geometry of a macro-tile shown on Fig.~\ref{fig-macrotile-utm-n} 
can be realized. This completes the construction of a self-similar aperiodic tile set.
Now, it is not hard to verify that the constructed tile set (i)~allows a tiling of the plane, and (ii)~is self-similar.

The construction described above works well for all large enough zoom factors $N$. In other words, for all large enough $N$ we get a self-similar tile set $\tau_N$, and the tilings for all these $\tau_N$ have very similar structure, with macro-tiles as shown in Fig.~\ref{fig-macrotile-utm-n}. 
Technically, the program $\pi$ (simulated by the universal Turing machine) now takes as its input a tuple of six strings of bits: the bit strings of length $k=k(N)$ representing the four macro-colors of a macro-tile,  the binary expansion of the zoom factor $N$, and  its own text.   This program  checks whether the given strings are coherent, i.e., whether the given quadruple of macro-colors in fact represents a quadruple of colors of one tile in our self-similar tile set $\tau_N$
(corresponding to the given value $N$ of the zoom factor).

The presented construction of a self-simulating tile set provides a proof of Theorem~\ref{thm-berger}, see a comment at the end of Section~\ref{ss:2-1}.
Indeed, if a tile set $\tau$ simulates itself with a zoom factor $N$, then by definition every $\tau$ tiling can be uniquely split into $N\times N$ macro-tiles.
Since these macro-tiles are isomorphic to the tiles of $\tau$, the grid of macro-tiles can be  uniquely split into macro-macro-tiles of size $N^2\times N^2$.
The macro-macro-tiles are obviously also isomorphic to the $\tau$-tiles, so they also can be batched in macro-tiles of higher rank, and so on.
We obtain a hierarchical structure of macro-tiles of rank $k=1,2,\ldots$, see Fig.~\ref{fig-self-similar-macrotile}.

For  discussion of the technique of self-simulating tilings and a motivation behind it we refer the reader to \cite{drs}.
 In what follows, we extend and generalize this construction step by step,  and  then apply it to prove much stronger statements.

\begin{figure}
\centering
\includegraphics[scale=0.52]{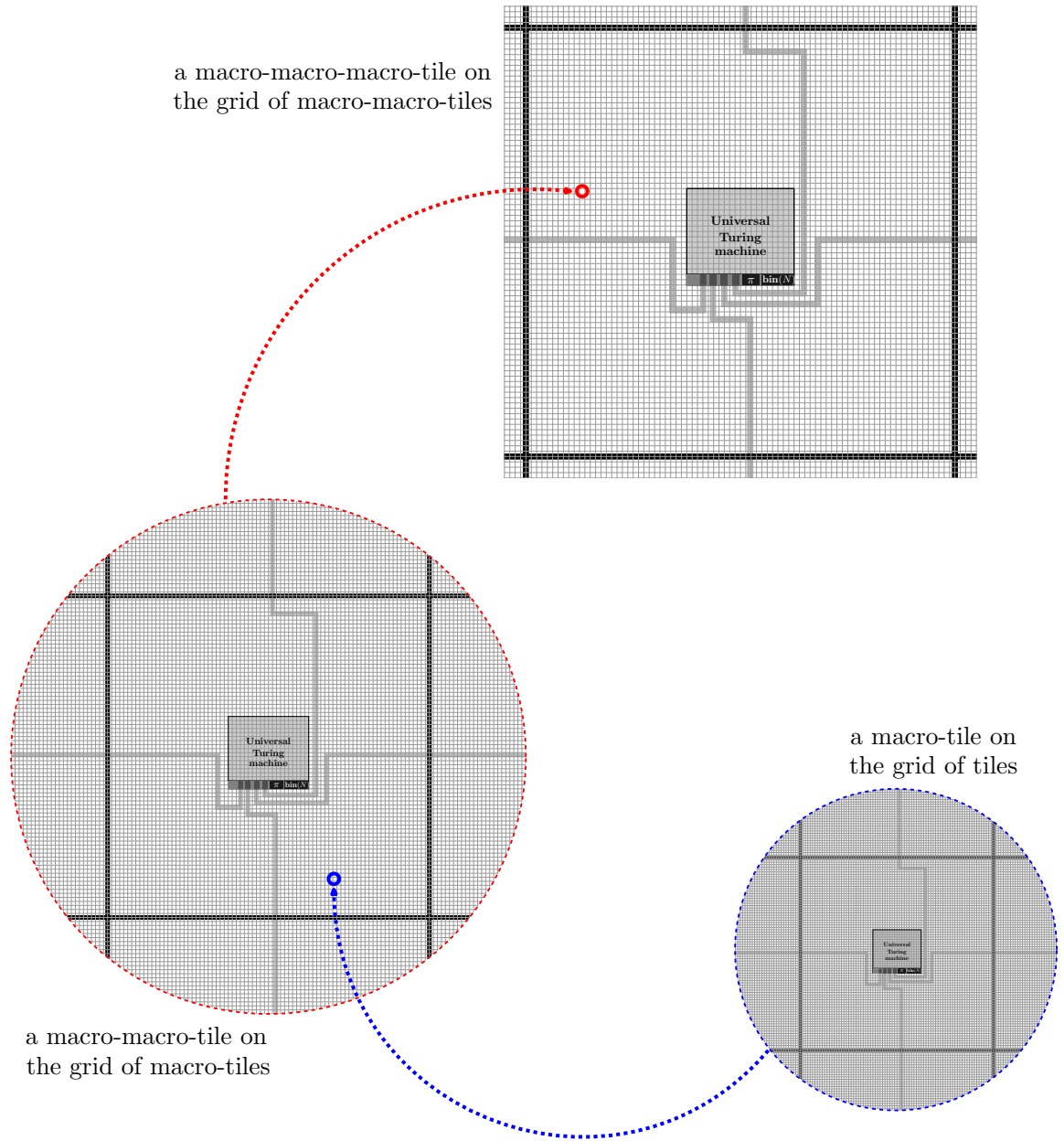} 
\caption{Hierarchical structure of macro-tiles. The level-$k$ macro-tiles are blocks in the level-$(k+1)$ macro-tiles, 
 the level-$(k+1)$ macro-tiles are blocks in level-$(k+2)$ macro-tiles, etc.
On all levels of the hierarchy, the structure of the macro-tiles is pretty much the same.}\label{fig-self-similar-macrotile}
\end{figure}

\smallskip

\emph{Notation.}
We now introduce some useful terminology. In a hierarchical structure of macro-tiles, if a level-$k$ macro-tile $M$ is a cell in a level-$(k+1)$ macro-tile $M'$, we refer to  $M'$ as the \emph{father} of $M$.
We refer to the level-$(k+1)$ macro-tiles neighboring $M'$ as the \emph{uncles} of $M$.

\subsection{A more flexible construction: The choice of the zoom factor}\label{ss-variable-zoom}

For a large class of sufficiently ``well-behaved'' sequences of integers $N_k$, we can construct a family of tile sets $\tau_k$ ($i=0,1,\ldots$)  such that each $\tau_{k-1}$ simulates the next $\tau_{k}$ with the zoom factor $N_k$ (and, therefore, $\tau_0$ simulates each $\tau_k$ with zoom factor $L_k = N_1\cdot N_2 \cdots N_{k}$).

The idea is to reuse the basic construction from the previous section and vary the sizes of the macro-tiles (the zoom factors) on the different levels of the hierarchy. While in the basic construction the macro-tiles (built of $N\times N$ tiles), the macro-macro-tiles (built of $N\times N$ macro-tiles), 
the macro-macro-macro-tiles (built of $N\times N$ macro-macro-tiles), and so on, behave in exactly the same way, in the revised construction the behavior of  a level-$k$ macro-tile depends on $k$. We want to have macro-tiles built of $N_1\times N_1$ ground-level  tiles, macro-macro-tiles built of $N_2\times N_2$ macro-tiles, macro-macro-macro-tiles built of $N_3\times N_3$ macro-macro-tiles, and so on. In this construction, the level-$k$ macro-tiles will  be isomorphic to the tiles of $\tau_k$, and the idea of ``self-simulation'' should be understood less literally.

To implement this idea, we need only a minor revision of the construction from the previous section.
Similar to our basic self-simulation construction, each tile of $\tau_k$  ``knows'' its coordinates modulo $N_k$ in the tiling:  the colors on the left and on the bottom sides should involve $(i,j)$, the color on the right side should involve $(i+1\mod N_k, j)$, and the color on the top side involves $(i, j+1 \mod N_k)$.
Consequently, every $\tau_k$-tiling can be uniquely split into blocks (macro-tiles) of size $N_k\times N_k$, where the coordinates of the cells range from $(0,0)$ in the bottom-left corner to $(N_k-1,N_k-1)$ in the top-right corner, similarly to Fig.~\ref{fig-0}. Again, intuitively, each macro-tile of level $k$ ``knows'' its position in the corresponding macro-tile of level $(k+1)$.  For each $k$, the  $N_k\times N_k$-macro-tile (built of tiles  $\tau_k$) should have the structure shown in Fig.~\ref{fig-macrotile-utm-n}, with communication wires, a computation zone, and auto-referential computation inside.

The difference with the basic construction is that now the computation simulated by a level-$k$ macro-tile gets, as an additional input, the value $k$, and the zoom factor $N_k$ is computed as a function of $k$. In what follows, we always assume that $N_k$ can be easily computed given the binary expansion of $k$ (say, in time $\poly(\log N_k)$).

Technically, we assume now  that the first line of the computation zone contains the following \emph{fields} of the \emph{input data}:
 \begin{itemize}
 \item[(i)] the program of a Turing machine $\pi$ that verifies whether a quadruple of macro-colors corresponds to a valid macro-color,
 \item[(ii)] the binary expansion of the integer rank $k$ of this macro-tile (the level in the hierarchy of macro-tiles),
 \item[(iii)] the bits encoding the macro-colors: each macro-color involves  the position inside the father macro-tile of rank $(k+1)$ (two coordinates modulo $N_{k+1}$) and 
 $O(1)$ bits of the supplementary information assigned to  the macro-colors.
 \end{itemize}
 Note that now the zoom factor is not provided explicitly as one of the input fields. Instead,  we have the binary expansion of $k$, so that a Turing machine can \emph{compute} the value of $N_k$, see Fig.~\ref{fig-macrotile-utm-k}.
 The difference from Fig.~\ref{fig-macrotile-utm-n} is that the computation in the macro-tile of rank $k$ gets as an input the index $k$ instead of the universal zoom factor $N$.
 \begin{figure}
\centering
\includegraphics[scale=0.62]{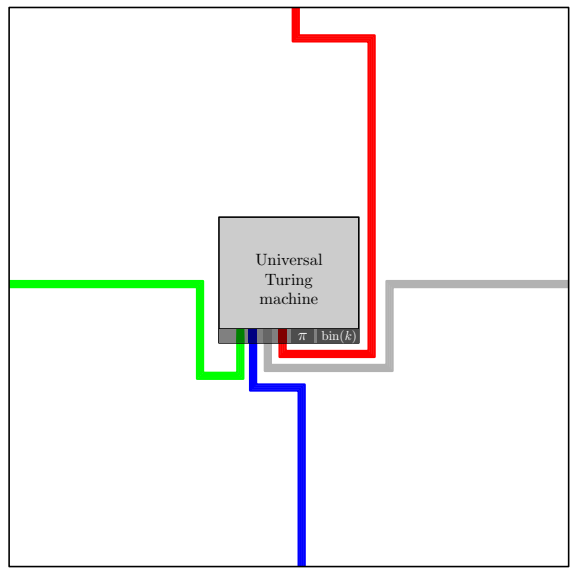}
\caption{A macro-tile of level $k$. The computation zone represents the universal Turing machine that simulates a program $\pi$, which gets as input the binary codes of  four macro-colors, the binary expansion  of the level $k$,  and the text of $\pi$ itself.}\label{fig-macrotile-utm-k}
\end{figure}
 
As before, we require that the simulated computation terminates in an accepting state if the macro-colors of the macro-tile form a valid quadruple
(if not, no correct tiling can be formed).
The simulated computation guarantees that the macro-tiles of level $k$ are isomorphic to the tiles of $\tau_{k+1}$. 

Note that on each level
$k$ of the hierarchy, we simulate in macro-tiles a computation of one and the same Turing machine $\pi$. Only the inputs for this machine
(including the binary expansion of the rank $k$) vary from level to level. 

This construction works well if $N_k$ does not grow too slowly (so that the level-$k$ macro-tiles have enough room to keep the binary expansion of $k$)
yet not too fast (so that the computation zone in the level-$k$ macro-tiles can handle elementary arithmetic operations with $N_{k+1}$).
In what follows we assume that $N_k = 3^{C^k}$ for some large enough constant $C$. 
 
The growing zoom factor $N_k$ allows embedding some \emph{payload} in the computation zone: some ``useful'' computation that has nothing to do with self-simulation but affects the properties of a tiling. 
More precisely, we require that the program $\pi$ (whose simulation by the Universal Turing machine is embedded in each macro-tile)
performs two different tasks: its primary job is to perform the checks of consistency for the four macro-colors of the corresponding  macro-tile, 
as explained above;
the secondary job is to run some specific ``useful'' algorithm $\cal A$.
In each proof based on this technique we use a particular algorithm $\cal A$ 
(which is explicitly hardwired in the program $\pi$ and, therefore, implicitly embedded in the constructed tile set).
If necessarily, this algorithm  may access as an input the  macro-colors of the corresponding macro-tile (they are available
in the computation zone). 
\emph{A priori}, the computation of $\cal A$ can be infinitely long.
We assume that in each macro-tile  the simulation of $\cal A$ is performed within the limits of the allocated space and time  
(the size of the ``computation zone''), and the simulation is aborted when the Universal Turing machine runs out these limits.
Though all macro-tiles (on all levels of the hierarchical structure) simulate one and the same algorithm $\cal A$, the available space
depends on the rank of a macro-tiles.
Since the zoom factor grows with the rank, on each subsequent level of the hierarchy of macro-tiles we can allocate 
to this secondary computation  more and more space and time.
 
 \label{disc:payload}
 
 \begin{rmrk}
 The presented construction of a self-simulating tile set is enough to prove the existence of an SFT where each configuration is  non-computable (the result known from \cite{nonrecursive1,nonrecursive2}), for details, see \cite{drs}. 
 \end{rmrk}

\section{Quasiperiodic self-simulating SFT}\label{s:quasiperiodic-sft}

In this section, we revise once again the  construction of a self-simulating tiling and enforce the property of quasi-periodicity or minimality. In particular, this construction will give a new proof of Theorem~\ref{thm-alexis-nicolas}. To implement this construction, we have to superimpose some new properties of a self-simulating tiling.
 
 \subsection{Supplementary features:  Constraints that can be imposed on the self-simulating tiling}
 
The tiles involved in our self-simulating tile set (as well as all macro-tiles of each rank) can be classified into three types:
 \begin{itemize}
 \item[(a)] the ``skeleton'' tiles, which keep no information except for their coordinates in the father macro-tile (the white area in Fig.~\ref{fig-macrotile-utm-k}; each of these tiles looks like the tile shown in Fig.~\ref{fig-0}):   these tiles work as building blocks of the hierarchical structure;
 \item[(b)] the ``communication wires,'' which transmit the bits of the macro-colors from the borderline of the macro-tile to the computation zone
 (the colored lines in Fig.~\ref{fig-macrotile-utm-k}; each of these tiles looks like the tiles shown in Fig.~\ref{fig-1-wires});
 \item[(c)] the tiles of the computation zone (intended to simulate the space-time diagram of the Universal Turing machine, the gray area in Fig.~\ref{fig-macrotile-utm-k}). 
 \end{itemize}
 Each pattern that includes only ``skeleton'' tiles (or ``skeleton'' macro-tiles of some rank $k$) reappears infinitely often in all homologous positions inside all macro-tiles of higher rank. Unfortunately, this property is not true for the patterns that involve  the ``communication zone'' or the ``communication wires.''  Thus, the basic construction of a self-simulating tiling does not imply the property of quasiperiodicity.
 To overcome this  obstacle we need  several new technical tricks.  

First of all, we impose several restrictions on our construction of a self-simulating tiling. 
These restrictions in themselves do not make the tilings quasiperiodic, but they simplify the upcoming revision of the construction.
More specifically, we  enforce the following additional properties (p1)--(p4) of a tiling, with only a minor modification of the construction.

\smallskip
\noindent
\textbf{(p1)} In our basic construction, each macro-tile contains a computation zone of size $m_k$, which is \emph{much} less than the size of the macro-tile $N_k$.  In what follows, we  need to reserve free space in a macro-tile, in order  to insert  $O(1)$ (some constant number) of copies of each $(2\times 2)$-pattern  from the computation zone (of this macro-tile), right above the computation zone. This requirement is easy to meet.
We assume that the size of a level-$k$ macro-tile (measured in blocks that are themselves macro-tiles of level $k-1$) is $N_k\times N_k$,
 and the computation zone in this macro-tile is  $m_k\times m_k$ for $m_k=\poly(\log N_k)$.
Therefore, we can reserve an area of size $\Theta(m_k)$ right above the computation zone, which is free of ``communication wires'' or any other functional gadgets, see the ``empty'' hatched area in Fig.~\ref{fig-macrotile-hatched-area}. So far, this area consisted  of only skeleton tiles; in what follows (Section~\ref{ss-3-2} below), we will use this zone to place some new nontrivial elements of the construction.
\begin{figure}
\begin{center}
\includegraphics[scale=0.62]{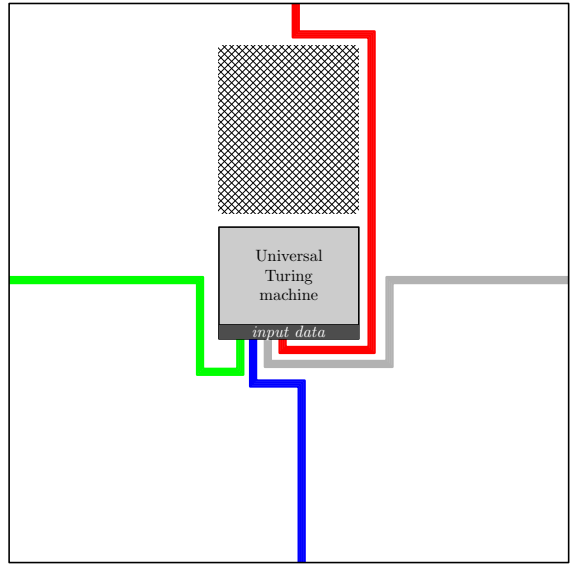}
\caption{The ``free'' area reserved above the computation zone.}\label{fig-macrotile-hatched-area}
\end{center}
\end{figure}

\smallskip
\noindent
\textbf{(p2)}  We require that the tiling inside the computation zone   satisfies the property of $2\times2$-\emph{determinacy}. That is, if we know all of the  colors  on the borderline of a $2\times2$-pattern inside the computation zone (i.e., a tuple of $8$ colors), then we can reconstruct the four tiles of this pattern. Again, we do not need  any new ideas to implement this property. It is not hard to see that this requirement is met if we represent the space-time diagram of a Turing machine in a natural way (see the discussion on p.~\pageref{remark-on-determ-turing-machine-diagram}). 

\smallskip
\noindent
\textbf{(p3)}  The communication channels in a macro-tile   (the wires that transmit information from the macro-color on the borderline of this macro-tile to the bottom line of its computation zone) must be isolated from each other. The gap between every two wires must be greater than two cells,  as shown in Fig.~\ref{fig-macrotile-isolated-wires}. In other words, each group of cells of size $2\times 2$  can touch at most one communication wire. Since the number of  wires in a level-$k$ macro-tile is only $O(\log N_{k})$, we have enough free space to lay the ``communication cables'' maintaining the required safety gap, so this constraint is easy to satisfy.
\begin{figure}
\begin{center}
\includegraphics[scale=0.47]{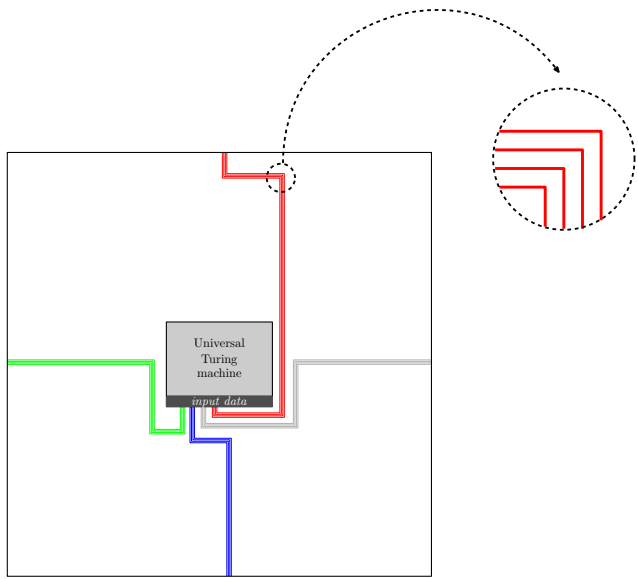}
\caption{The wires in the ``communication cables''  (shown in red, blue, green, and gray in the figure) are separated by a gap (shown in white), so that the  distance (measured in tiles) between every two wires is greater than two.}\label{fig-macrotile-isolated-wires}
\end{center}
\end{figure}

\smallskip
\noindent
\textbf{(p4)}  In our construction, the macro-colors of  a level-$k$ macro-tile are encoded by bit strings of length $r_k = O(\log N_{k+1})$.
In the previous section we only  assumed that this encoding is somewhat ``natural'' and easy to handle. So far, the choice of encoding was of small importance: we only required that some natural manipulations with macro-colors could be implemented in polynomial time.
\label{property:p4}

We now add another (seemingly artificial) condition.
We have decided that each macro-color is  encoded in a string of $r_k$ bits. We require now
that each bit in this encoding takes both values $0$ and $1$ \emph{quite often}.  More precisely,
we require that for each $i=1,\ldots,r_k$  there are \emph{quite many} macro-tiles where the $i$th 
bit  of encoding of the  top (bottom, left, right)  macro-color is equal to $0$, and there are 
\emph{quite  many}  other macro-tiles where  the $i$th bit of this encoding is equal to $1$. 
In what follows we specify what the words \emph{quite often} and \emph{quite many} mean in 
this context.

Technically, we use the following property: for every position $s=1,\ldots, r_k$ and for every $i=0,\ldots,N_{k+1}-1$ we require that
\begin{itemize}
\item there exists  $j_0$ such that the $s$th bit in the top, left, and right macro-colors of the level-$k$ macro-tile at the positions $(i,j_0)$  in the level-$(k+1)$ father macro-tile  are equal to  $0$, and
\item there exists  $j_1$ such that the $s$th bit in the top, left, and right macro-colors of the level-$k$ macro-tile at the positions $(i,j_1)$  in the level-$(k+1)$ father macro-tile  are equal to $1$.
\end{itemize}
There are many (more or less artificial) ways to implement  this constraint. For example, we may subdivide the array of $r_k$ bits encoding a macro-color  into three equal zones of size $r_k/3$ and require that for each macro-tile only one of these three zones contains the ``meaningful'' bits, and the two other zones contain only zeros and ones respectively; we require then that the ``roles'' of these three zones cyclically permute as we go upwards along a column of macro-tiles,
see Fig.~\ref{fig-excessive-encoding-of-macrocolors}.
\begin{figure}
\centering
\includegraphics[scale=0.60]{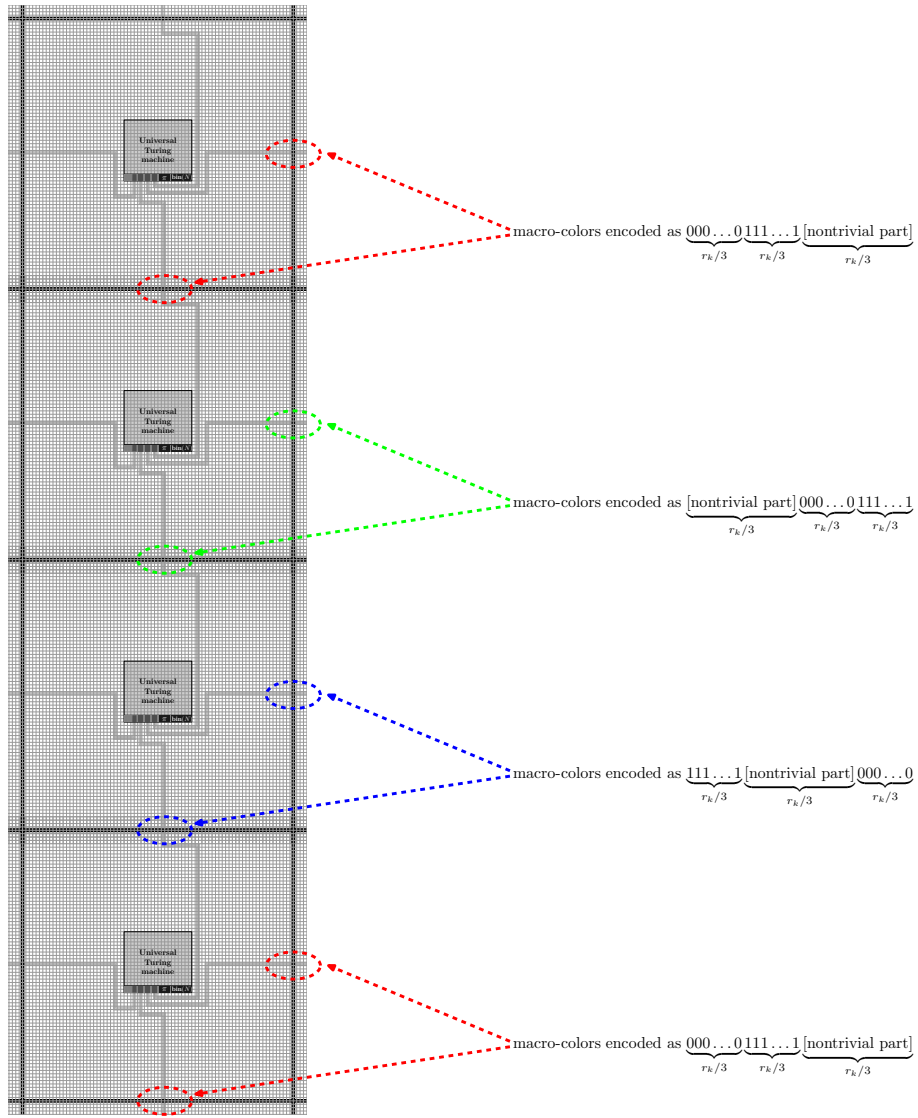}
\vspace{-10pt}
\caption{Encoding of macro-colors. The $r_k$ bits of the code are split into three blocks of size $r_k/3$. One of them consists of all $0$s, another of all $1$s, and only the third one contains a nontrivial binary code. The roles of the three  blocks  change cyclically  
from one macro-tile to another when we move upwards.}\label{fig-excessive-encoding-of-macrocolors}
\end{figure}

\subsection{Enforcing minimality}\label{ss-3-2}

To achieve the property of minimality of an SFT, we should guarantee  that every finite pattern that appears \emph{once} in at least one tiling must also appear in \emph{every large enough square} in every tiling.  
In a tiling with a hierarchical structure of macro-tiles  each finite pattern  can be covered by at most four macro-tiles  (by a $2\times2$-pattern)  of an appropriate rank. Hence, to guarantee the property of minimality, it is enough to show that every $(2\times 2)$-block of macro-tiles of any rank $k$ that appears in at least one $\tau$-tiling actually reappears in this tiling in every large enough square.  Let us classify all $(2\times 2)$-block of macro-tiles (by their position in the father macro-tiles of higher rank) and discuss what revisions of the construction are needed.

\emph{Case 1: Skeleton tiles.}
For a $(2\times 2)$-block of four ``skeleton'' macro-tiles of level $k$, there is nothing to do. Indeed, in  our construction we have exactly the same blocks with  every vertical translation by a multiple of  $L_{k+1}$ (we have there a similar block of level-$k$ ``skeleton'' macro-tiles contained in some other  macro-tile of rank $(k+1)$). 

\emph{Case 2: Communication wires.}
Let us consider the case when a $(2\times 2)$-block of level-$k$ macro-tiles involves a part of a communication wire. Due to property (p3) we may assume that only one wire is involved. The bit transmitted by this wire is either $0$ or $1$; in either case, due to property (p4),
we can find another similar $(2\times 2)$-block of level-$k$ macro-tiles (at the same position within the father macro-tile of rank $(k+1)$ 
and with the same bit included in the communication wire)  in  every macro-tile of level $(k+2)$. 
In this case we can find a duplicate of the given block with a vertical translation of size $O(L_{k+2})$.  

\emph{Case 3: Computation zone.} Now we consider the most difficult case: when a $(2\times 2)$-block of level-$k$ macro-tiles touches the computation zone.
In this case we cannot obtain the property of quasiperiodicity for free, and so we have to make one more  modification  to our general construction of a self-simulating tiling.
\begin{figure}
\centering
\includegraphics[scale=0.6]{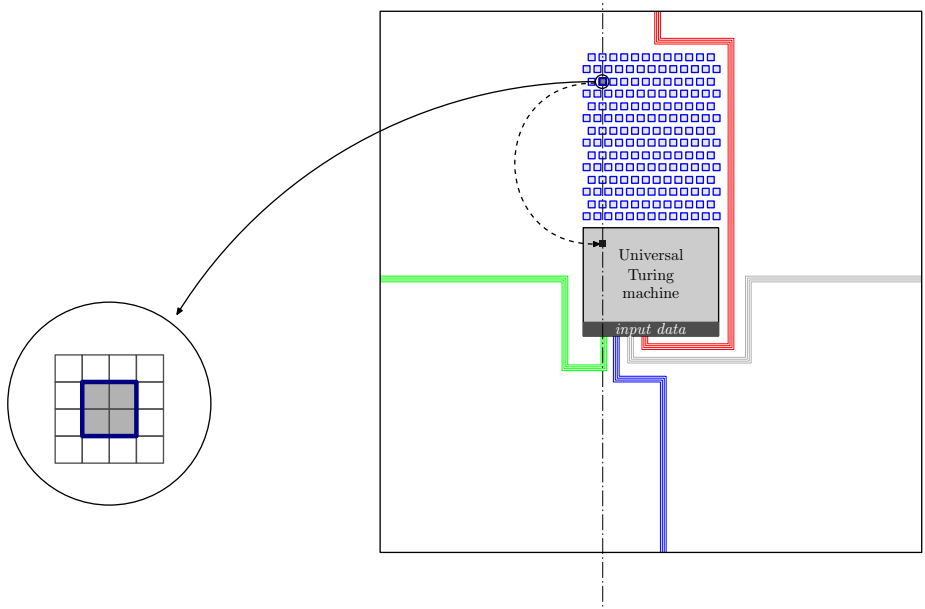}
\vspace{-10pt}
\caption{Positions of the \emph{diversification slots} for patterns from the computation zone.}\label{fig-macrotile-with-slots-detailed}
\end{figure}

 Note that for each $2\times2$-window that touches the computation zone  of a macro-tile, there are only $O(1)$ ways to tile them correctly. For each possible position of a $2\times2$-window in the computation zone and for each possible filling of this window by tiles, we reserve a special $2\times2$-\emph{slot} in a macro-tile, which is essentially a block of size $2\times2$ in the ``free'' zone of a macro-tile. We refer to this gadget as a \emph{diversification slot}. These slots will enforce the property of ``diversity'':  for every small pattern that \emph{could} appear in the computation zone, we will guarantee that it \emph{must} appear in the corresponding diversity slot in every macro-tile of the same rank. 
 
 The  diversification slots must be placed   far away from the computation zone and from all communication wires. We prefer to place every  diversification slot  in the same vertical stripe as the ``original'' position of this block, as shown in Fig.~\ref{fig-macrotile-with-slots-detailed} (this property of vertical alignment will be used in Section~\ref{s-subdynamics}).
 We have enough free space to place all necessary  diversification slots, due to property (p1). We  define the neighbors around each  diversification slot in such a way that only one specific $(2\times 2)$-pattern  can patch it (here we use the property (p2)).  
 
In our construction, the tiles around this  slot ``know'' their real coordinates in the bigger macro-tile, while the tiles inside the  diversification slot do not (they ``believe'' they are tiles in the computation zone, while in fact they belong  to an artificial isolated  ``diversity preserving'' slot far outside of any real computation), see Fig.~\ref{fig-macrotile-with-slots-detailed} and Fig.~\ref{fig-slots-detailed}.
 The frame of the  diversification slot consists of 12 ``skeleton'' tiles (the white squares in Fig.~\ref{fig-slots-detailed}); they form a slot  that involves inside a $(2\times 2)$-pattern extracted from the computation zone (the gray squares in Fig.~\ref{fig-slots-detailed}). 
In the picture, we show the ``coordinates'' encoded in the colors on the sides of each tile.  
Note that the colors of the bold lines (the blue lines between the white and gray tiles and the bold black lines between the gray tiles) should contain some information beyond the coordinates---these colors involve the bits used to simulate a space-time diagram of the universal Turing machine. 
In this picture, the ``real'' coordinates of the bottom-left corner of this slot are $(i+1,j+1)$, 
while the ``natural'' coordinates of the  pattern inside the  diversification slot (when this pattern appears in the computation zone) are $(s,t)$. 

\begin{figure}
\centering
\includegraphics[scale=0.40]{p16.mps}
\caption{A  diversification slot for a $2\times2$-pattern from the computation zone.}\label{fig-slots-detailed}
\end{figure}
 
 We choose the positions of the  diversification slots in the macro-tile so that the coordinates can be computed by some simple algorithm in time polynomial in $\log N_k$. We require that all  diversification slots be detached from each other in space, so they do not damage the general structure of the ``skeleton'' tiles building the macro-tiles.

\smallskip

Now it is not hard to see that for the revised tile set,  every pattern that appears \emph{at least once}  in \emph{at least one} tiling must in fact appear in \emph{every large enough} pattern in \emph{every} tiling.
Thus,  the revised construction of a self-simulating tiling guarantees that 
\begin{itemize}
\item every tiling is aperiodic (the argument from the previous section remains valid), and 
\item every pattern that appears at least once in at least one configuration must appear in every large enough square in every tiling.
\end{itemize}
Hence, our construction implies Theorem~\ref{thm-alexis-nicolas}.

\section{Quasiperiodicity and non-computability}\label{s:noncomputability}

So far, we have used the method of self-simulating tilings to combine the properties of quasiperiodicity (and even minimality) 
 and aperiodicity. Now we are going to go  further and
combine quasiperiodicity with non-computability. 
In this section, we extend the construction discussed above and prove  Theorem~\ref{thm1} and Theorem~\ref{thm1-bis}.
To this end, we extensively use  the technique of a  self-simulating  tiling with a \emph{variable zoom factor}, introduced in Section~\ref{ss-variable-zoom}. 
As we mentioned above, we assume that the size of a macro-tile of rank $k$ is equal to $N_k\times N_k$, for  $N_k=3^{C^k}$, $k=1,2,\ldots$, and the size of the computation zone $m_k$ grows as $m_k = \poly(\log N_k)$.

We take as the starting point the generic scheme of a quasi-periodic self-simulating tiling explained in the previous section, and adjust it with some new features. 
From now on we  require that all macro-tiles of rank $k$ contain, in their computation zone, the prefix (e.g., of length $\lceil \log k \rceil $) of some infinite sequence
 $X=x_0x_1x_2\ldots$
 We want that all macro-tiles of rank $k$ contain the same prefix $x_0x_1x_2\ldots k_{\lceil \log k \rceil}$, so we will embed these bits in the macro-colors of all  macro-tiles.  To make things more pictorial, we require also that these bits be provided in the bottom line of the computation zone as one supplementary input field, as shown in Fig.~\ref{fig-macrotile-comp-zone}. But we stress again that the bits  $x_0x_1x_2\ldots k_{\lceil \log k \rceil}$ make part of each macro-color. So in Fig.~\ref{fig-macrotile-comp-zone} these bits are repeated five times: they appear four times implicitly as a part the codes of macro-colors (greed, blue, gray, and red fields in the picture) and one more time explicitly (the rightmost part of the computation zone). The computation hidden in each macro-tile should verify that these five copies of $x_0x_1x_2\ldots k_{\lceil \log k \rceil}$ provided  to the computation zone are identical.

Since every two neighboring macro-tiles must have matching macro-colors, we can guarantee every two neighboring level-$k$ macro-tiles contain one and the same prefix $x_0x_1x_2\ldots k_{\lceil \log k \rceil}$ (this remains true even for neighboring level-$k$  macro-tiles that belong to different father macro-tiles of level $(k+1)$). Thus, the construction implies that in each valid configuration all level-$k$ macro-tiles involve one and the same prefix  $x_0x_1x_2\ldots k_{\lceil \log k \rceil}$. 
 
 We also need coherence between bits $x_i$ embedded in macro-tiles of different ranks (the string embedded in a macro-tile of higher rank should extend the string embedded in macro-tiles of lower rank).
 So we suppose that the computation embedded in the computational zones of macro-tiles verifies whether the input data are coherent, i.e.,  that the bits embedded in each of the four macro-tiles match the bits $x_0x_1x_2\ldots k_{\lceil \log k \rceil}$ given explicitly in this new input field.
 \begin{figure}
\centering
\includegraphics[scale=0.65]{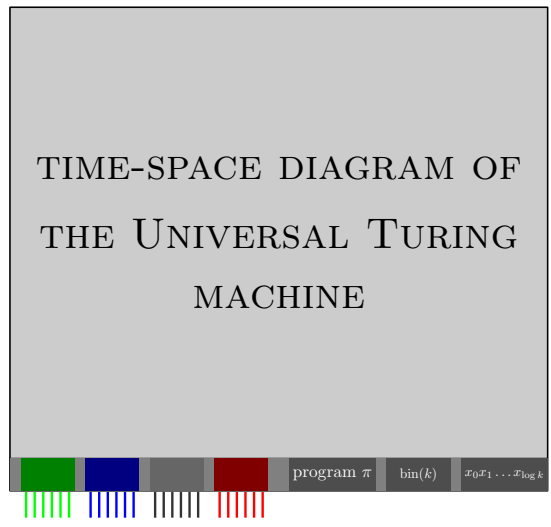}
\caption{The \emph{computation zone} of a macro-tile in the revised construction. The input data consist of the codes of the macro-colors, the simulated program $\pi$, 
the binary expansion of the rank $k$, and the first $\log k$ bits of the embedded sequence $X=x_0x_1x_2\ldots$}\label{fig-macrotile-comp-zone}
\end{figure}
Using the usual self-simulation, we can guarantee that the bits of $X$ embedded in a macro-tile of rank $k+1$ extend the prefix embedded in a macro-tile of rank $k$. Since the size of the computation zone increases as a function of $k$, the entire tiling of the plane determines  an infinite sequence of bits $X$ (whose prefixes are encoded in the macro-tiles of all ranks).

\begin{rmrk}
The new feature allows embedding an infinite sequence $X$ in a tiling. This embedding is \emph{highly distributed} in the following sense: we can extract the first $\log k$ bits of the sequence from any level-$k$ macro-tile.  Note that the tile set does not determine the embedded sequence $X$
\textup(different tilings of the same tile set can represent different sequences $X$\textup). However we are going to control the class of sequences $X$ that
can be embedded in a tiling.
\end{rmrk}

Since the embedded sequence $X$ is not uniquely defined by the tile set, we lose the property of \emph{minimality}
(different $\tau$-tilings involve different embedded sequences $X$, and therefore different  finite patterns). 
However, we still have the property of \emph{quasiperiodicity}. Indeed, every valid tiling contains a well-defined embedded sequence $X$.
Let us fix one infinite sequence $X$ and restrict ourselves to the class of tilings ${\cal T}(X)$ that represent this specific sequence.
Then, all level-$k$ macro-tiles in every tiling in ${\cal T}(X)$  involve (one and the same) prefix of $X$.
Now the argument from the previous section, repeated word for word, gives the following property.

\smallskip

\noindent
\emph{Every $(2\times 2)$-block of level-$k$ macro-tiles that appears at least once in at least one tiling in ${\cal T}(X)$
must reappear   in every large enough pattern of every tiling in ${\cal T}(X)$.
}

\smallskip

\noindent
It follows immediately that every tiling in this construction is quasiperiodic.

When the class of embedded sequences of $X$ is not restricted, this construction is of no interest. 
It becomes meaningful if we can enforce some special properties of the embedded sequence $X$.

As we explained in Section~\ref{section:2.2}, we can include in the computation performed by each macro-tile 
a simulation of some specific algorithm $\cal A$ (see the discussion on p.~\pageref{disc:payload}). We may
require that $\cal A$ verifies some particular property of the embedded sequence $X$. Each macro-tile
accesses a prefix of the embedded sequence, so we may test some computable properties of 
the available prefix of $X$. As usual, the computation  in each macro-tile is limited by the size of the computation zone.
However,  we may assume that the  space and time allocated to the computation
grow with the levels of the hierarchy.

We start with a rather simple example of a property of $X$ that can be implemented. We can guarantee that for every tiling the embedded  sequence $X$  is 
\emph{not computable}. 
This is not completely trivial: our aim is to implement a \emph{computable} procedure which guarantees that a sequence of bits $X$ is \emph{not computable}.
More technically, we need to find a computable set of constraints (forbidden finite patterns) such that each $X$ that satisfies these constraints
(i.e., contains none of the forbidden patterns) is not computable. From the computability theory we know many examples of sets with the required properties
(co-recursively enumerable sets with no computable points, see, e.g., \cite[Proposition~V.5.25]{odifreddi}).

Our construction can be combined with any example of a co-recursively enumerable non-empty set  $A\subset \{0,1\}^\mathbb{N}$   with no computable points (i.e., each sequence $x\in A$ must be incomputable). We can embed such a set $A$  in a quasi-periodic self-simulating tiling, so that the resulting SFT satisfies the statement of Theorem~\ref{thm1}. 
For the sake of self-containedness, we prefer to be more specific and describe below in some detail one particular construction of  such a set.

\smallskip

\emph{Short digression: a co-recursively enumerable set of bit sequences  with no computable points.}
Arguably the simplest classic construction of a set with this property is based on a pair of  recursively inseparable  enumerable sets. 
Let us remind that there exists 
a pair of disjoint recursively enumerable sets $S_1, S_2\subset \mathbb{N}$ such that there is no decidable ``separator'' $W$ satisfying $S_1\subset W$ and $S_2\subset \mathbb{N}\setminus W$ (see, e.g., \cite[Section~II.2]{odifreddi} and \cite[Section~2.4]{shen-vereshchagin-book}).
In order to  implement the outlined plan, we embed in the computation zone of macro-tiles an algorithm $\cal A$  that performs the following job: it enumerates two non-separable enumerable sets (on each level $k$ we run these two enumerations for the  number of steps that fits the  computation zone available in a macro-tile of rank $k$). Then,  $\cal A$ checks that the bits of $X$ represent a separator between these two sets. Technically, in each macro-tile 
the computation verifies that these (partially) enumerated sets are indeed separated by the given prefix of $X$.  
In every valid tiling, the embedded sequence $X$ must pass the checks performed in all macro-tiles on all levels of the hierarchy. 
Hence, this $X$ must be a separator between two non-separable enumerable sets. Therefore, this sequence $X$ must be non-computable.
It follows that the tiling (which contain this sequence) must be also non-computable.

\smallskip

Combining all the ingredients together, we obtain a tile set $\tau$ that  enjoys two nontrivial properties: all $\tau$-tilings are non-computable and quasiperiodic. This gives a new proof of  Theorem~\ref{thm1}. 

Thus, we constructed a tile set $\tau$ such that all $\tau$-tilings are quasiperiodic and non-computable. Up to now, we could not say much more about the degree of unsolvability  (Turing degrees) of $\tau$-tilings. Now we are going to enhance this construction  by implementing some   more precise control
on the class of embeddable sequences $X$, and therefore on the class of possible 
Turing degrees of $\tau$-tilings.   We start with a proposition that  characterizes the no-go zones for this technique.

\begin{theorem} \label{p-turing-degrees}
(a)~Every SFT is effectively closed.
(b)~For every infinite minimal SFT $\cal S$, the class of the Turing degrees representable by configurations in  $\cal S$ is upper-closed: if there exists a $\tau$-tiling that has a Turing degree $T$, then every Turing degree $T'>T$ is  also represented by some $\tau$-tiling.
\end{theorem}
\begin{proof}
(a)~is trivial, (b)~is proven in~\cite{jeandel-vanier}.
\end{proof}

\begin{rmrk} Observe that we cannot guarantee that the Turing degrees of \emph{all} $\tau$-tilings are \emph{very} high. More specifically,
 we cannot guarantee that {all} $\tau$-tilings are \emph{not low} or 
\emph{not hy\-per\-im\-mune-free}. 
Indeed,   due to  the low basis theorems, for every tile set $\tau$, some $\tau$-tilings are not low and not hyperimmune-free.
\end{rmrk}

Theorem~\ref{thm1-bis} essentially claims  that 
 a class of Turing degrees which is not forbidden by Theorem~\ref{p-turing-degrees} (i.e., a class that is upwards closed and corresponds to an effectively closed set) can be implemented by a suitable tile set.

\begin{proof}[Proof of Theorem~\ref{thm1-bis}]
To prove this theorem, we again employ the  idea of embedding an infinite sequence $X$ in a tiling, and control more precisely the properties of the embedded sequence. 
Similarly to the construction discussed above, we require that all macro-tiles of rank $k$  involve  the same finite sequence of $\log k$ bits on their computation zone, which is understood as a prefix of $X$. We can guarantee that the prefix embedded in macro-tiles of rank $k$ is compatible with the prefix available to the macro-tiles of the next rank $(k+1)$. 

Further, since $\cal A$ is in $\Pi_1^0$, we can enumerate the (potentially infinite) list of patterns that should not appear in $X$. On each level, the macro-tiles run this enumeration for the available space and time (limited by the size of the computation zone available on this level), and verify that the discovered forbidden patterns do not appear in the prefix of $X$ accessible to the macro-tiles of this level. Since the computation zone becomes bigger and bigger with each level, the enumeration extends longer and longer. Thus, a sequence $X$ can be embedded in an infinite tiling if and only if this sequence does not contain any forbidden pattern (i.e., if this $X$ belongs to $\cal A$).

What are the Turing degrees of the tilings in the described tile set? In our tile set, every tiling is uniquely defined by the following information: the sequence $X$ embedded in this tiling and the sequences of integers $\sigma_h, \sigma_v$ that specify the shifts (the vertical and the horizontal ones) of macro-tiles of each level relative to the origin of the plane. Indeed, on each level $k$ we split the macro-tiles of the previous rank into blocks of size $N_k\times N_k$. These blocks make  level-$k$ macro-tiles, and there are $N_k^2$ ways to choose the grid of horizontal and vertical lines that define this splitting.
Given these sequences $\sigma_h, \sigma_v$ and an $X\in{\cal A}$, we can reconstruct the entire tiling. 
It remains to note that $\sigma_h$ and $\sigma_v$ can be absolutely arbitrary.
Thus, the Turing degree of a tiling is the Turing degree of  $(X,\sigma_h, \sigma_v)$, which can be an arbitrary degree not less than $X$. That is, the set of degrees of tilings is exactly the closure of ${\cal A}$, i.e., the set of all $Y$ that are not less than some $X\in{\cal A}$. So we get the statement of  Theorem~\ref{thm1-bis}.
\end{proof}

\section{Transitive version of the Hochman--Meyerovitch theorem about possible entropies of SFT}

\begin{figure}
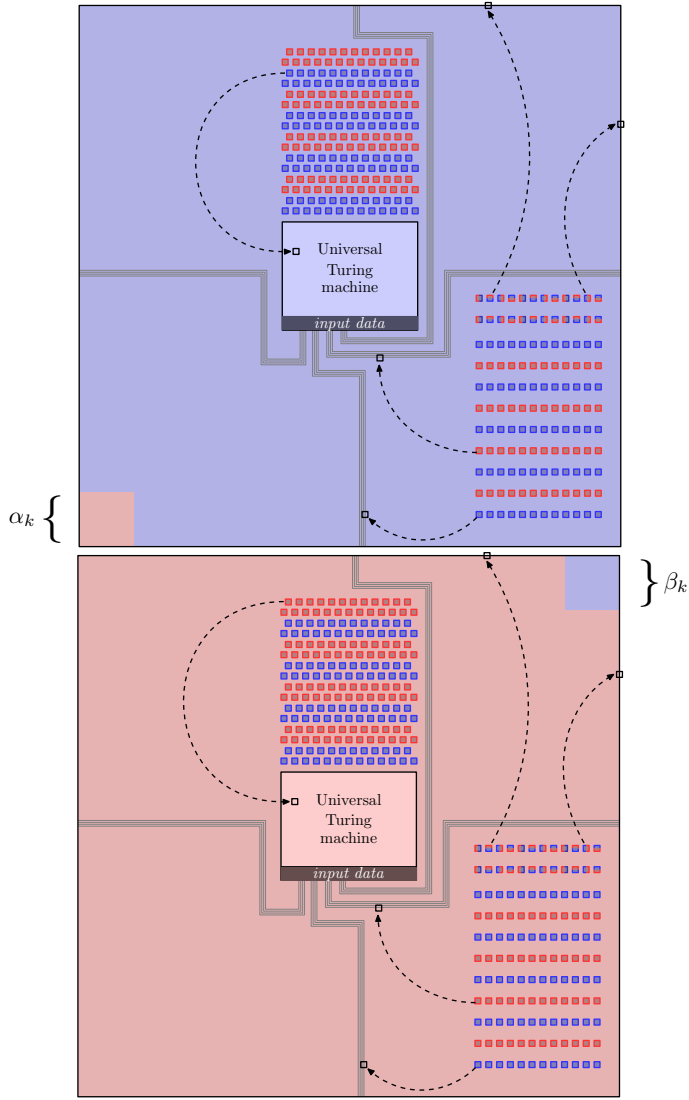

\begin{center}
\includegraphics[scale=0.6]{p24.mps}
\rule{1.65cm}{0cm}
\includegraphics[scale=0.6]{p23.mps}
\caption{A {blue} macro-tile of rank $k$ consists of mostly {blue} blocks of rank $(k-1)$, with only an $(\alpha_k  \times \alpha_k)$ corner of {red} blocks. For a {red} macro-tile the situation is dual: most blocks are {red},  with a $(\beta_k \times\beta_k)$ corner of {blue} blocks.  
In comparison with Fig.~\ref{fig-macrotile-with-slots-detailed}, here we have  diversification slots for $2\times2$-pattern not only  in the computation zone but also in the communication wires and on the borderline of the macro-tile.
In both types of macro-tiles (red and blue) we keep the  diversification slots with both red and blue versions of the ``cloned'' patterns.}\label{fig-blue-red-macro-tiles}
\end{center}
\end{figure}

In this section, we  prove  Theorem~\ref{t:transitiv-hochman-meyerovich}.
As usual, we consider only the case  $d=2$.
Before proceeding with the proof,  we establish the following proposition.

\begin{proposition}\label{p:transitiv-hochman-meyerovich}
For every integer $d>1$ and 
for every  nonnegative right recursively enumerable real $h$ there exists a tile set $\tau$ split into two disjoint subsets, $\tau=\tau_R\cup \tau_B$
\textup(we will say that the tiles in $\tau_R$ are {red} and the tiles in $\tau_B$ are {blue}\textup) that admits a tiling of the plane, and
\begin{itemize}
\item the set of $\tau$-tilings is a minimal shift,
\item $\limsup$ of the fraction of {red} tiles in globally admissible $n\times n$ patterns 
\textup(i.e., in correctly tiled squares of size $n\times n$ that can appear in a tiling of the plane\textup) is equal to $h$.
\end{itemize}
\end{proposition}
\begin{proof}
We take as the starting point the construction of a tile set from Section~\ref{s:quasiperiodic-sft}, which enforces a tiling with growing zoom factors and guarantees the minimality of the corresponding SFT. We upgrade it in the following way.

\emph{Red and blue tiles and macro-tiles.} We want that each tile and each macro-tile (of each rank) ``knows'' its color,  which can be red or blue. On the ground level, we basically make two copies of every tile (one copy is red and another is blue). Below  we specify the local constraints imposed on the ``colors''  assigned to the neighboring tiles.
We also slightly revise the structure of the macro-tiles (which implicitly requires, of course, some minor modification of the ground level tiles).

In the new construction, every macro-tile can be ``red'' or ``blue.'' First of all, \emph{the color of a macro-tile} (represented by a single bit of information) is written as a new  input field in the bottom line of the computation zone. (Thus, we add one more input field besides the macro-colors and  the binary expansion of macro-tile's rank, similarly to Fig.~\ref{fig-macrotile-comp-zone}.)  The color of the cells composing a level-$k$ macro-tile
(as usual, the cells are macro-tiles of level $(k-1)$) are subject to the following rule:
 \begin{itemize}
 \item the cells in the bottom-left corner of size $\alpha_k  \times \alpha_k$ are always {red} 
 (whatever is the ``color'' of the macro-tile in whole);
 \item  the cells in the top-right corner of size $\beta_k \times \beta_k$ are always {blue} 
 (whatever is the ``color'' of the macro-tile in whole);
 \item the color in  diversification slots (that guarantee quasiperiodicity) can be red or blue, see below;
 \item the other cells (skeleton cells as well as communication wires and the computation zone) are  {red}  in a red macro-tile
 and {blue} in a blue macro-tile. 
 \end{itemize}
 The structure of ``red'' and ``blue'' macro-tiles is show in Fig.~\ref{fig-blue-red-macro-tiles}.
 
 \begin{rmrk}
(a) The parameters $\alpha_k$ and $\beta_k$ are small enough so that the ``red'' and ``blue'' corners in a macro-tile neither intersect the computation zone nor the communication wires. We fix the values of these parameters later.
 
(b) The coherence of the ``color'' of most cells in a macro-tile is guaranteed by local rules:  neighboring cells inside a macro-tile  must have the same ``color'' (except for the two special corners mentioned above and the  diversification slots).
 
(c) Since the cells in the computation zone ``know'' their colors, it is easy to guarantee that the color of the macro-tile provided as an input of the Turing machine (simulated  in the computation zone) is coherent with the color of the cells in this macro-tile.
 \end{rmrk}
 
It remains to explain the policy for the color of  the diversification slots (which reproduce all $(2\times 2)$ patterns from the computation zone). For each of the $(2\times 2)$-patterns touching the computation zone, we reserve twice as many slots as we did in  Section~\ref{s:quasiperiodic-sft}. For each of these configurations, we prepare both ``red'' and ``blue''  clones, no matter what the real color of the entire macro-tile is. 

In addition to the  diversification slots duplicating the patterns in the 
computation zone,  we  embed, in the macro-tile, similar slots (with both possible colors) for all $(2\times 2)$-patterns that intersect the communication wires, see Fig.~\ref{fig-blue-red-macro-tiles}. (In a level-$k$ macro-tile, we have only $O(\log N_k)$ wires, and the length of each wire is at most $O(N_k)$. Hence, we have enough free space to place all necessary  diversification slots.)

Last, we also add similar  diversification slots to the $(2\times 2)$-patterns intersecting the borderline of the macro-tile. Note that neighboring cells  that belong to different macro-tiles  may have different colors. So in this case we build  diversification slots  with multi-colored $2\times2$-patterns inside. 
Since the length of the borderline of a macro-tile is $O(N_k)$, we have enough free space in a macro-tile to make an isolated  diversification slot for each of these patterns. This family of gadgets concludes the construction of the new tile set.

\smallskip

\noindent
\emph{Claim 1.} The set of all tilings for the described tile set is a minimal SFT.

\begin{proof}[Proof of Claim 1.]
We need to show that every finite pattern that appears in at least one tiling must appear in every tiling (in every large enough area). As usual,
we profit from the structure of a hierarchical self-simulating tiling: it is enough to prove this property for patterns built of 
$(2\times 2)$-group of level-$k$ macro-tiles. For a $(2\times 2)$-pattern that touches either the computation zone or a communication  wire or the borderline of  a macro-tile of level $(k+1)$, this property is simple to establish:   we can find a clone of such a pattern inside the corresponding  diversification slot (which exists in every macro-tile of rank $k+1$). 

For a $(2\times 2)$-pattern that consists of only ordinary skeleton cells, things are somewhat trickier. We cannot say that exactly the same 
pattern can be found at the homologous position in \emph{every} other macro-tile of rank $(k+1)$. In fact, we only find the same pattern
at the same position in macro-tiles of rank $(k+1)$ with the same color. Thus, it remains to explain why in every tiling, for every $k$, there exist
both red and blue macro-tiles of rank $k+1$. Here we use the existence of red and blue corners in our macro-tiles. Indeed, by construction,
in every  macro-tile of rank $(k+2)$ there exist some red cells and some blue cells (which are themselves macro-tiles of rank $(k+1)$). This argument 
works for any reasonable choice of $\alpha_k$ and $\beta_k$: we only need to assume that $\alpha_k$ and $\beta_k$ are 
strictly positive for all $k$.
\end{proof}

\smallskip

\noindent
\emph{Claim 2.} The values of $\alpha_k$ and $\beta_k$ can be chosen so that
(a) these parameters are computable (as functions of $k$) in the space and time available in the computation zone of macro-tiles of rank $(k-1)$,
and (b) $\limsup$ of  the fraction of {red} tiles in the globally admissible blocks of size $n\times n$ is equal to $h$.

\begin{proof}
To prove  property (b), it is enough to estimate the fractions of red and blue tiles 
inside a macro-tile (of growing rank). It is clear that the fraction of red tiles in level-$k$ red and blue macro-tiles 
(denote this by $\nu_R(k)$, and that for blue by $\nu_B(k)$)
can be computed recursively by

\[
\begin{array}{rcl}
\nu_R(k) &:=& \left[%
	\begin{array}{c}
\text{fraction of red level-$(k-1)$ macro-tiles}\\
\text{in each red level-$k$ macro-tile}
\end{array}
\right]\times \nu_R(k-1) \\ \\
                &&{} +  \left[%
                \begin{array}{c}
                \text{fraction of blue level-$(k-1)$ macro-tiles}\\
                \text{ in a red level-$k$ macro-tile}
                \end{array}
                \right]
                \times \nu_B(k-1),\\
                \\ \\
\nu_B(k) &:=& \left[%
	\begin{array}{c}
	\text{fraction of red level-$(k-1)$ macro-tiles}\\
	\text{in a blue level-$k$ macro-tile}
	\end{array}
	\right]\times \nu_R(k-1) \\ \\
                &&{} +  \left[%
                	\begin{array}{c}
		\text{fraction of blue level-$(k-1)$ macro-tiles}\\
		\text{in a blue level-$k$ macro-tile}
		\end{array}
		\right]\times \nu_B(k-1).                
\end{array}
 \]

 Recall that the zoom factor $N_k$ grows very fast (as $3^{C^k}$). The fractions of red and blue level-$(k-1)$ macro-tiles in a red level-$k$ macro-tile
 depend on the choice of $\alpha_k$ and $\beta_k$, and on the fraction of diversification slots (by construction, for each slot the ``color'' of the involved pattern is uniquely defined, and does not depend on the color of the entire macro-tile). The fraction of a macro-tile occupied by the diversification slots is only $O\big(\frac{\log N_{k+1}}{N_k}\big)$, and the choices of $\alpha_k$ and $\beta_k$ are under our control.
  
The construction works well for a very wide range  of parameters. 
For example, we can set $\alpha_k=1$ (or $\alpha_k=\mathrm{const}$ for any other constant) for all $k$, and vary $\beta_k$  between $1$ and, 
say,  $N_k/10$ (as a function of $k$).   

Note that if  $\alpha_k=\mathrm{const}$ and $\beta_k=\mathrm{const}$ for all $k$, then 
the value of $\nu_R(k)$ converges to some (computable) limit  in the interval $(0,1)$. Moreover, we can make this limit arbitrarily close to $1$ 
by choosing a large enough constant $C$ in the definition of the zoom factor $N_k=3^{C^k}$.

It remains to explain how to reduce the limit of  $\nu_R(k)$  to the given right recursively enumerable number $h<1$.  
To this end we increase the values of $\beta_k$ (for some indices $k$). 
There is a difficulty:   we cannot compute the exact value of $h$ in finite time, unless $h$ is rational.  Moreover,  in general, we cannot
even compute an $\epsilon$-approximation of $h$ for every given precision $\epsilon$. 
We only know that this number is right recursively enumerable. That is, by definition, there is an algorithm
which never halts and  enumerate an infinite sequence of rational numbers that converge to the limit $h$  from the right. We embed
in the computation zone this algorithm and simulate it within the space and time available on the computation zone
(the simulation is aborted when the computation runs out of the allocated space and time, see the discussion on p.~\pageref{disc:payload}).
As the level of a macro-tile grows, the size of the computation zone becomes greater, and we have more and more space
and time to simulate this computation.  Hence, on each next level of the hierarchy of macro-tiles we obtain better and better approximations 
of $h$. (We have no tool to estimate the precision of the current approximation, but we know that the limit of the enumerated sequence
is equal to $h$.)

Thus, each level-$k$ macro-tile obtains some approximation of $h$ (from above). 
Accordingly,  given the current approximation of $h$, we compute a suitable value of $\beta_k$,
so that $\nu_R(k)$ converges to  $h$ as $k$ tends to infinity.
\end{proof}

Combining Claim~1 and Claim~2 completes the proof of the theorem.

\end{proof}

\begin{proof}[Proof of Theorem~\ref{t:transitiv-hochman-meyerovich}]
At first we consider the case $h<1$.
Let us take the SFT $\cal S$ from the proof of Proposition~\ref{t:transitiv-hochman-meyerovich}. Now we make two copies of each red proto-tile and denote the resulting shift by ${\cal S}'$. We claim that ${\cal S}'$ has entropy $h$. Indeed, the hierarchical structure of macro-tiles gives no contribution to the entropy. Indeed, in the initial SFT $\cal S$ every level-$k$ macro-tile of size $L_k\times L_k$ can be reconstructed given its color (red or blue) and the macro-colors on its borderline, which requires only $O(L_k)$ bits of information. Thus, the positive entropy of the new shift ${\cal S}'$ results  from the choices between two copies of each red tile (at every position where such a tile is used). Hence, the entropy is equal to the $\limsup$ of the density of red tiles, which is guaranteed to be $h$.  

In case $h> 1$, we can superimpose the same construction with a trivial shift (without any local constraints) with an alphabet of size $2^{\lfloor h \rfloor}$.

\smallskip

Let us prove that ${\cal S}'$ is transitive.
The minimality of ${\cal S}$ means that {every} configuration in this shift contains all globally admissible patterns.  Let us fix any configuration  $x\in{\cal S}$. 
Now we want  to make a choice for each position with a red tile in $x$ 
so that the resulting configuration  involves  all  possible instantiations  of every   globally admissible pattern. 
In fact, this property is true
with probability $1$ if we choose an instance of each red tile in $x$ at random. Indeed,   
for every finite pattern  (with finitely many red tiles) we can find  in $x$ infinitely many disjoint  copies.
Hence, if we choose a version of every red tile at random, then with probability $1$ we  obtain infinitely
many copies of every instantiation of every finite pattern. 
\end{proof}

\begin{rmrk}
Theorem~\ref{t:transitiv-hochman-meyerovich} guarantees the existence of a transitive SFT with a given right recursively enumerable entropy.  The construction in the proof implies that this SFT enjoys also a (weak) version of irreducibility: every two globally admissible patterns can be combined in one common configuration; moreover, every two globally admissible patterns can be positioned rather close to each other (roughly speaking, every two globally admissible macro-tiles of level $k$ can be placed inside one and the same globally admissible macro-tile of level $(k+2)$). However, the relative arrangement of two globally admissible patterns must be coherent with the global hierarchical grid of macro-tiles. For example, the vertical and the horizontal translations of a pattern $P_1$ with respect to another pattern  $P_2$ in a common infinite configuration  modulo the first rank zoom factor $N_1$ is uniquely defined (by the this pair of patterns).
Gangloff and Sablik in \cite{gangloff-sablik} proposed a construction with a somewhat stronger property of irreducibility, where the relative arrangement of any two globally admissible patterns is rather flexible.
\end{rmrk}

\section{On subdynamics of co-dimension $1$ for self-si\-mu\-la\-ting SFT}\label{s-subdynamics}

In Section~\ref{s:noncomputability} we used a sort of  embedding  of one-dimensional sequences in a two-dimensional SFT.
That embedding was highly distributed: the first $\log k$ bits of the sequence embedded in a configuration could be found in every macro-tile of rank $k$ of this configuration.
In this section we discuss a different way of embedding a (bi-infinite) one-dimensional sequence in two-dimensional configurations of an SFT; 
this version of embedding is less distributed and more local. With this technique we will be able to control the subdynamics of a two-dimensional shift,
and as a result we will prove Theorems~\ref{thm-main}--\ref{thm-main-min}.

\subsection{The general scheme of letter delegation}\label{subsection:delegation}

We are going to embed a bi-infinite sequence $\mathbf{x} = (x_i)$ over an alphabet $\Sigma$ into our tiling.
To this end we assume that each individual $\tau$-tile ``keeps in mind'' a letter from $\Sigma$ that propagates without change in the
vertical direction.  Formally speaking, a letter from $\Sigma$ should be a part of the top and bottom colors of every $\tau$-tile 
(the letters assigned to both sides of a tile must be equal to each other),  see Fig.~\ref{fig-tile-with-embedded-letter} 
and   Fig.~\ref{fig-tiling-with-embedded-letters}.  
\begin{figure}
\begin{center}
\includegraphics[scale=0.75]{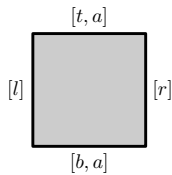}
\caption{A tile propagating a letter $a\in \Sigma$ in the vertical direction. Formally speaking, this tile is a quadruple of \emph{colors},
the left side has color  $[l]$, the right side has color $[r]$, the top and the bottom sides have colors $[t,a]$ and $[b,a]$, respectively.
The colors for the top and bottom sides involve a letter from $\Sigma$.
We allow only tiles where the colors of the top and bottom sides involve one and the same letter.  }\label{fig-tile-with-embedded-letter}
\end{center}
\end{figure}
We want to guarantee that a $\Sigma$-sequence  can be embedded in a $\tau$-tiling if and only if this sequence belongs to a fixed given effective 
shift $\cal A$.   (We postpone for a while the discussion of \emph{quasiperiodicity} of this embedding.)

The general plan is to ``delegate'' the factors of the embedded sequence to the computation zones of macro-tiles,
where these factors can be validated (that is, the simulated Turing machine  can verify whether these factors do not contain any forbidden subwords). 
By using tilings with growing zoom factors, we can guarantee that  the size of the computation zone of a $k$-rank macro-tile grows with $k$. 
So we have at our disposal the computational resources required to run all necessary validation tests on the embedded sequence.
It remains to organize the propagation of the letters of the embedded sequence to the ``conscious memory'' 
(the computation zones) of the macro-tiles of all ranks. In what follows we explain how this propagation is organized. 

 \emph{The zone of responsibility  of a macro-tile.}\label{letter-delegation}
In our construction,  a macro-tile of level $k$ is a square of size $L_k\times L_k$, with 
$L_k=N_1\cdot N_2\cdot \ldots\cdot N_{k}$ (where $N_i$ is the zoom factor on level $i$ of the hierarchy of macro-tiles).
We say that a level-$k$ macro-tile is \emph{responsible} for the letters of the embedded sequence $\mathbf{x}$
assigned to the columns of the (ground level) tiles of this macro-tile as well as  to the columns of macro-tiles of the same rank 
on its left and on its right. 
That is, the \emph{zone of responsibility} of a level-$k$ macro-tile is a factor of length $3L_k$
from the embedded sequence, see Fig.~\ref{fig-responsib}. (The zones of responsibility of  two vertically aligned macro-tiles are the same; the zones of responsibility of  two horizontally neighboring macro-tiles overlap.)

\begin{figure}
\begin{center}
\includegraphics[scale=0.70]{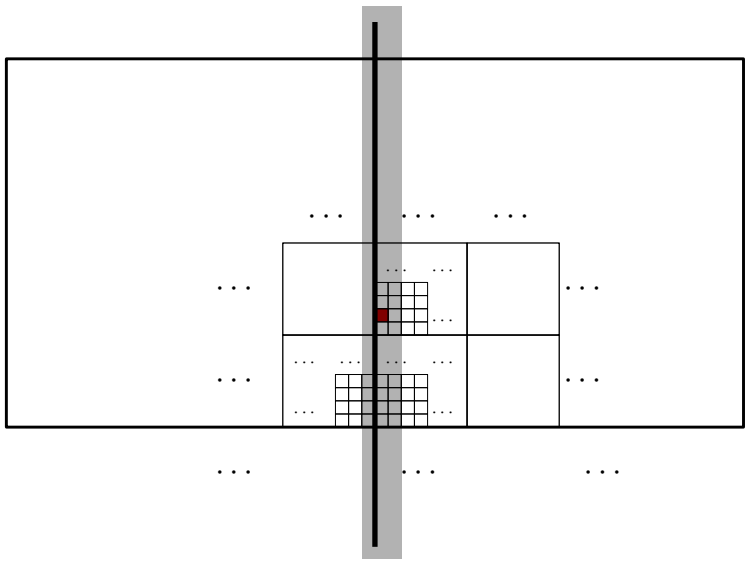}
\caption{The zone of responsibility (the gray vertical stripe)  of a macro-tile (the red square) is three times wider than the macro-tile itself.}\label{fig-responsib}
\end{center}
\end{figure}

\emph{Letter assignment:}\label{letter-assignment}
The computation zone of a level-$k$ macro-tile (of size $m_k\times m_k$) is too small to contain all the letters from its zone 
of responsibility. So we require that the computation zone obtains as an input a (short enough) chunk of letters from its zone 
of responsibility. Let us say that it is a factor of length $l_k := \log \log L_k$ from the stripe of $3L_k$ columns constituting the zone
of responsibility of this macro-tile. We will say that this chunk is \emph{assigned} to this macro-tile.

The infinite stripe of vertically aligned level-$k$ macro-tiles share the same zone of responsibility.  However, different macro-tiles 
in such a stripe will obtain different assigned chunks. The choice of the assigned chunk varies from $0$ to $(3L_{k}-l_k)$. Therefore, we need
to choose for each level-$k$ macro-tile a position of a factor of length $l_k$ in its zone of responsibility of length $3L_k$.  
This choice is quite arbitrary.
Let us say, for definiteness, that for a macro-tile $M$ of rank $k$
the first position of the assigned chunk (in the stripe of length $3L_k$) is defined as the vertical position of $M$ in the father macro-tile
of  rank $(k+1)$ (taken modulo $(3L_{k}-l_k)$). 

\begin{rmrk} 
We have chosen zoom factors $N_k$ growing doubly exponentially in $k$, so $N_{k+1}\gg 3L_k$. Hence, every chunk of length $l_k$ from
a stripe of width $3L_k$ is assigned to some \textup(actually, to infinitely many\textup) of the macro-tiles ``responsible''  for these $3L_k$ letters.
\end{rmrk}

\begin{rmrk}  
Since the zones of responsibility of neighboring level-$k$ macro-tiles overlap by more than $l_k$, every finite factor of length $l_k$ in the embedded sequence $\mathbf{x}$  is assigned  to some level-$k$ macro-tile \textup(even if the $l_k$ columns containing the letters of this factor are not covered by any single level-$k$ macro-tile and touch two horizontally neighboring level-$k$ macro-tiles\textup), see Fig.~\ref{fig-extended-resp-zone}.
\begin{figure}
\centering
\includegraphics[scale=0.30]{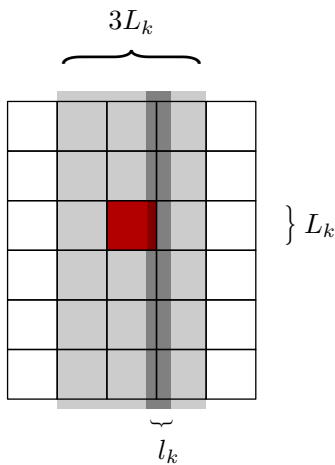}
\caption{The macro-tile of size $L_k$ (shown in red) is \emph{responsible} for the vertical stripe of width $3L_k$ shown in light gray (three times wider than the macro-tile itself). Such a macro-tile can handle a factor of length $l_k$ of the embedded sequence that corresponds to the group of columns that touch this macro-tile as well as its  neighbor on the right or on the left (an example is shown in dark gray).}\label{fig-extended-resp-zone}
\end{figure}
\end{rmrk}

\emph{Implementing the letter assignment by self-simulation.} In the \emph{letter assignment} paragraph above, we formulated several requirements: how the data should be be propagated from the ground level (individual tiles) to level-$k$ macro-tiles.  That is, for each level-$k$ macro-tile $\cal M$ we specified which chunk of the embedded sequence should be a part of the data fields on the computation zone of  $\cal M$. So far we have not explained how this propagation can be implemented, i.e., how the assigned chunks can arrive at the high-level data fields.  
Now, we will explain how to implement the required scheme of letter assignment in a self-simulating tiling. 
 Technically, we append to the input data of the computation zones of macro-tiles some supplementary data fields:
 \begin{itemize}
 \item[(iv)] a block of $l_k$ letters from the embedded sequence assigned to this macro-tile,
 \item[(v)] three blocks of bits of $l_{k+1}$ letters  of the embedded sequence assigned to this father macro-tile,
 and two uncle  macro-tiles (the left and the right neighbors of the father),
 \item[(vi)] the coordinates of the father macro-tile in the ``grandfather'' (of rank $(k+2)$).
 \end{itemize}
(In other words, the first line of the computation zone still looks similar to Fig.~\ref{fig-macrotile-comp-zone}, 
but now it contains more input data in the first line of the computation zone.)
Speaking informally, the computation in each level-$k$ macro-tile must check the consistency of the data in fields (iv), (v) and (vi). 
That is, if some letters from the fields 
(iv) and (v) correspond to the same vertical column in the zone of responsibility, then these letters must be equal to each other. Also,
if a level-$k$ macro-tile plays the role of a cell in the computation zone of the level-$(k+1)$ father, it should check the consistency of
its (v) and (vi) with the bits displayed in  the father's computation zone.  Lastly, we must ensure the coherence of  the fields (v) and (vi)
for each pair of neighboring level-$k$ macro-tiles; so this data should also be a part of the macro-colors (with a minor exception, see
Remark~\ref{r:field-vi-exception} below).

Note that the data from the ``uncle'' macro-tiles are necessary to deal with the letters from the columns that physically  belong to the
neighboring macro-tiles. So the consistency of the fields (v) is imposed also on  neighboring  level-$k$ macro-tiles that belong to 
different  level-$(k+1)$ fathers (the borderline between these  level-$k$ macro-tiles is also the borderline between their fathers).

The coherence of fields (iv), (v), (vi) on every level of the hierarchy implies that for each macro-tile the content of  field~(iv)
fairly represents  the assigned factor of the embedded sequence.
If we want to be formal, this statement can be proven by induction on the level $k$ of the hierarchy of macro-tiles. 
Indeed, by construction, 
the chunk of $l_{k}$ letters assigned to a level-$k$ macro-tile $M$  is ``known'' to the ``children''  and ``nephews'' of this macro-tile --- 
to those level-$(k-1)$ macro-tiles that constitute $M$  itself and its ``brothers'' on the left and on the right; these 
level-$(k-1)$ macro-tiles make sure that their father's (or uncle's) chunk of $l_{k}$  letters is consistent with their own chunks of $l_{k-1}$ letters.
The letters assigned to the  level-$(k-1)$ macro-tiles must be consistent with the letters assigned to their own children, which
are level-$(k-2)$ macro-tiles, and which in turn keep on their computation zones the assigned chunks of $l_{k-2}$ letters, and so on.
On the ground level we explicitly make sure that the assigned letters are consistent with the symbols of the embedded sequence.

This scheme works properly since the chosen zoom factors $N_k$ grow very fast, so that $N_{k}\gg L_{k-1}$. This guarantees
for every level-$k$ macro-tile $M$ that each factor of the embedded sequence of length $l_{k-1}$ that appears in the zone of responsibility of $M$,
is assigned to several level-$(k-1)$ macro-tiles $M'$ that are children or nephews of $M$. So every letter assigned to $M$ is validated by 
macro-tiles of the previous level.  

This procedure of validation of letter assignment is even redundant --- each letter from the zone of responsibility of a level-$k$ macro-tile $M$ 
is validated at once by  \emph{many} level-$(k-1)$ macro-tiles $M'$ inside $M$. The delegation still works correctly 
even if we exclude some level-$(k-1)$ macro-tiles 
from the procedure of validation of the data of their level-$k$ father. 
More precisely, the construction still works properly, 
if (a)~the excluded macro-tiles occupy only 
$O(1)$ successive positions in each column, and (b) the non-excluded macro-tiles constitute 
a connected component and, therefore, are coherent with each other. Thus, in what follows we may revise our construction
and exclude, for example, the cells of the \emph{communication wires} from  validation of their father's data,
see Remark~\ref{r:field-vi-exception} below.

The  computation that verifies the coherence of the data in fields (iv)--(vi) is pretty simple. It can be performed in polynomial time, and 
the required revision of the construction fits the usual constraints on the parameter (the size of the computation zone in a level-$k$ macro-tile
is $\poly\log(N_k)$). For a detailed discussion of the hierarchical
schema of ``letter delegation''  we refer the reader to  \cite[Section~7]{drs}).

\begin{rmrk}\label{r:field-vi-exception}
As mentioned above, the defined construction  is somewhat excessive:  
to ensure the correct ``information propagation'' from the level $(k-1)$   to the  level $k$ of the hierarchy, 
we do not need to keep the content of the auxiliary fields~(v) and~(vi) in  \emph{each} macro-tiles  of level $(k-1)$.

We  take advantage of this observation and make a minor (seemingly artificial) revision of our construction.
We assume that the content of the field~(vi) is \emph{empty} for the macro-tiles that play  in their fathers the role of a communication wire, 
as well as the neighbors of the  communication wires. 
(Every macro-tile ``knows'' its position in the father macro-tiles, so it knows whether it is a communication wire or not). 
Observe that the macro-tiles with a non-empty field~(vi) form a connected component in their father macro-tile, so they must be 
coherent with each other.
The purpose of this modification  will become clear in the proof of Theorem~\ref{thm-main} and Theorem~\ref{thm-main-min}
(see Remark~\ref{r:field-vi-exception-in-use} below).
\end{rmrk}

\emph{Concluding remarks: Testing against forbidden factors.}
To guarantee that the embedded sequence $\mathbf{x}$ contains no forbidden patterns, each level-$k$ 
macro-tile should allocate some  part of its computation zone to enumerate (within the limits of available space and time) 
the forbidden pattern, and verify that the block of $l_k$ letters assigned to this macro-tile contains none of the  forbidden factors found.

The time and space allocated to  enumerating the  forbidden words  grow as  functions of $k$. To ensure that the embedded sequence
contains no forbidden patterns,  it is
enough to guarantee that each forbidden pattern is found by macro-tiles of high enough rank, and every factor of the embedded sequence
is compared (on some level of the hierarchy) with every forbidden factor.

Thus, we get a  construction of a two-dimensional tiling that simulates a given effective one-dimensional shift: for a given effective one-dimensional
shift $\cal A$, we can construct a tile set $\tau$ such that  the bi-infinite sequences that can be embedded in $\tau$-tilings are
exactly the sequences of $\cal A$. 
In the next sections we explain how to make these tilings quasiperiodic in the case when the simulated one-dimensional shift is also quasiperiodic.

\subsection{The letter delegation scheme combined with quasiperiodicity}\label{subsection:delegation-quasiperiodicity}\label{subsection:comb-lemmas}

In Section~\ref{s:quasiperiodic-sft} we described a very general construction of a self-simulating tile set, and showed that the corresponding SFT enjoys the properties of quasiperiodicity or even minimality. In the previous section we upgraded this construction and superimposed on the generic scheme of self-simulation   a new  technique:  the scheme of embedding in a tiling a  sequence from some effective one-dimensional shift. In general, the new ``upgraded'' SFT may lose the property of quasiperiodicity. To maintain it, we will need some additional effort. The following lemma is a useful technical tool: it helps  control the properties of quasiperiodicity and minimality of the tilings with embedded sequences.
 \begin{lemma}\label{lemma-clones-with-embedded-bits}
For a tile set defined in Section~\ref{subsection:delegation}, 
two globally admissible macro-tiles of rank $k$ are equal to each other if these macro-tiles
\begin{itemize}
\item[\textup{(a)}] contain the same bits in fields \rm{(i)--(vi)} in the input data on the computation zone,
   and
 \item[\textup{(b)}] the factors of the encoded sequence corresponding to the zones of responsibility   of these macro-tiles \textup(in the corresponding vertical stripes of width $3L_k$\textup) are equal to each other.
 \end{itemize}
 \end{lemma}
 \begin{proof}
 The proof is by induction on rank $k$. For a macro-tile of rank $1$, the statement follows directly from the construction. 
 In the inductive step, we are given a pair of macro-tiles $M_1$ and $M_2$ of rank  $(k+1)$ that hold identical data in fields (i)--(vi), and
 the  factors (of length $3L_k$) from the encoded sequences in the zones of responsibility of  $M_1$ and $M_2$  are also equal to each other.
We observe that the corresponding cells in $M_1$ and $M_2$
 (which are macro-tiles of rank $k$) contain the same data in their own   fields (i)--(vi), since  the communication wires of $M_1$ and $M_2$
 carry the same information bits, their computation zones represent exactly the same 
 computations, etc. 
Therefore, we can apply the inductive assumption.
 \end{proof}
 The statement of Lemma~\ref{lemma-clones-with-embedded-bits} is chosen in such a way  that the inductive proof is simple.
However, to apply this lemma, it is useful to separate the data contained  in a macro-tile into two parts:
 the data relevant to the construction of the hierarchical structure of macro-tiles,  and the symbols of the embedded sequence:
  \begin{corollary}\label{corollary-clones-with-embedded-bits}
 For a tile set defined in Section~\ref{subsection:delegation},  two globally admissible level-$k$ macro-tiles $M_1$ and $M_2$ are equal to each other if 
 the following conditions hold true
\begin{eqnarray}\label{lemma-clones-with-embedded-bits-a}
\left.
\begin{array}{l}
\bullet\ \text{\rm$M_1$ and $M_2$ have the same position \textup(modulo $N_{k+1}$\textup) with respect} \\
\text{\rm to their fathers, which are level-$(k+1)$ macro-tiles;} \\
\\
\bullet\  \text{\rm the fathers of $M_1$ and $M_2$ have the same position \textup(modulo $N_{k+2}$\textup) } \\
\text{\rm with respect to the grandfathers of $M_1$ and $M_2$, which are in turn }\\
\text{\rm level-$(k+2)$ macro-tiles;}\\
\\
\bullet\ \text{\rm if  $M_1$ and $M_2$ play the role of communication wires in their fathers, } \\
\text{\rm then these wires communicate the same value \textup(i.e,  both of them}\\
\text{\rm  communicate  $0$ or both of them communicate  $1$\textup)};\\
\\
\bullet\ \text{\rm  if  $M_1$ and $M_2$ are involved in the computation zone in their fathers, } \\
\text{\rm  then they contain identical finite patterns of the space-time diagram};\\
\end{array}
\right\}
 \end{eqnarray}
and 
\begin{eqnarray}\label{lemma-clones-with-embedded-bits-b}
\left.
\begin{array}{l}
\bullet\ \text{\rm the factors of length $3L_{k}$ of the embedded sequence for which  $M_1$} \\
\text{\rm  and $M_2$ are responsible are equal to each other; moreover, the factors    } \\
\text{\rm of length $3L_{k+1}$  of the embedded sequence for which the fathers of} \\
\text{\rm  $M_1$ and $M_2$ are responsible, are also equal to each other}. \\
\end{array}
\right\}
 \end{eqnarray}
 \end{corollary}
\begin{proof}
Conditions  \eqref{lemma-clones-with-embedded-bits-a} and  \eqref{lemma-clones-with-embedded-bits-b}
imply that the fields \rm{(i)--(vi)} in the input data on the computation zones of $M_1$ and $M_2$ are equal to each other.
Besides,  \eqref{lemma-clones-with-embedded-bits-b} implies that  $M_1$ and $M_2$ contain  in their zones of responsibility
the same factors of the embedded sequence. Thus, we can apply Lemma~\ref{lemma-clones-with-embedded-bits}.
\end{proof}
\begin{rmrk}
We made  Condition \eqref{lemma-clones-with-embedded-bits-b} very strict
(we could be less restrictive on the  symbols in the zones of responsibility of the fathers of  $M_1$ and $M_2$)
in order to simplify the future applications of this corollary.
\end{rmrk}

To control the property of  quasiperiodicity of self-simulating tilings, 
we will use two  simple  lemmas concerning quasiperiodic sequences.  
One of them (Lemma~\ref{lemma-quasiperiodic-times-periodic}, which is known from  \cite{avgust,salimov}) is purely combinatorial.
In the other one (Lemma~\ref{lemma-minimal-times-periodic}, which to the best of our knowledge is new), we combine the combinatorial
properties with an algorithmic twist.
For the sake of self-containedness, we give the proofs of  both of these lemmas in the next section.

\begin{lemma}
\label{lemma-quasiperiodic-times-periodic}
 Let $\mathbf{x}=(\ldots x_{-1} x_0 x_1 x_2 \ldots)$ be a bi-infinite  recurrent sequence,  $v=x_s x_{s+1}\ldots x_{s+N-1}$ be an $N$-letter factor  
in $\mathbf{x}$, and $q$ a positive integer.
Then there exists an integer $t>0$ such that another copy of $v$ appears in  $\mathbf{x}$ with a translation $q\cdot t$, i.e., 
 \begin{equation}\label{eq:lemma2}
  x_s x_{s+1}\ldots x_{s+N-1} = x_{s+qt}x_{s+qt+1}\ldots x_{s+qt+N-1}.
 \end{equation}
In other words, in  $\mathbf{x}$ there exists another instance of the same factor $v$ with a translation divisible by $q$.

Moreover, if $\mathbf{x}$ is quasiperiodic, then  for all integers $q$ and $N$ there exists an integer $L=L(\mathbf{x},N,q)$  such that 
the absolute value of  $(qt)$ in \eqref{eq:lemma2} can be chosen less than $L$.
\textup(Note that $L$ does not depend on $s$, i.e., it does not depend on a specific instance of the pattern $v$ in $\mathbf{x}$.\textup)
\end{lemma}
\begin{rmrk}
From the symmetry, it follows that a similar statement can be proven with an integer  $t<0$. Thus, we can find in $\mathbf{x}$ 
two copies of $v$, one on the left and the other one on the right of the originally given instance of this factor, both copies with translations divisible by $q$.
Clearly, we can iterate this procedure and obtain a bi-infinite (infinite to the left and to the right) sequence of copies of $v$, where each copy has a translation 
(with respect to the original instance of the factor) divisible by $q$.
\end{rmrk}
\begin{rmrk}\label{rmrk:dense-siblings}
The property of quasiperiodicity  guarantees by definition that there is a uniform bound for the gaps between neighboring appearances of 
each $N$-letter factor $v$  in $\mathbf{x}$. 
Lemma~\ref{lemma-quasiperiodic-times-periodic} claims that there is also a uniform bound for the gaps between neighboring appearances of 
each $N$-letter factor $v$  in $\mathbf{x}$, for appearances at positions that are congruent to each other modulo $q$.
\end{rmrk}

\begin{rmrk}
For higher dimensions, Lemma~\ref{lemma-quasiperiodic-times-periodic} can be generalized as follows.
 Let $\mathbf{x}$ be a  $d$-dimensional recurrent configurations on $\mathbb{Z}^d$ (over a finite alphabet $\Sigma$),  
 $v$ be a finite pattern  in $\mathbf{x}$, and $q$ be a positive integer.
Then in $\mathbf{x}$ there exists another instance of the same pattern $v$   such that  the translation between these two copies of 
$v$ is a non-zero vector $\bar t = (t_1,\ldots,t_d)$, where each component $t_i$ is divisible by $q$.
Moreover, if $\mathbf{x}$ is quasiperiodic, then the size of each $t_i $ can be bounded by some number $L$ that depends only on $\mathbf{x}$, $q$,  and $v$ \textup(but not on a specific instance of the pattern $v$ in $\mathbf{x}$\textup). 

Below we prove Lemma~\ref{lemma-quasiperiodic-times-periodic}  for $d=1$; our argument can be easily extended to the general case (for any $d>1$).
\end{rmrk}

In the next lemma we use the following notation.
For a configuration $\mathbf{x}$ (over some finite alphabet) we denote with ${\cal S}(\mathbf{x})$ the shift that consists of all configurations $\mathbf{x}'$ containing only patterns from $\mathbf{x}$. If a shift $\cal T$ is minimal, then ${\cal S}(\mathbf{x}) = {\cal T}$ for all configurations $\mathbf{x}\in {\cal T}$.
\begin{lemma}
\label{lemma-minimal-times-periodic}
(a) Let  $\cal T$ be an effective minimal shift. Then for every  $\mathbf{x}=(x_i)$ from ${\cal T}$ and  every periodic configuration $\mathbf{y}=(y_i)$,
the direct product $\mathbf{x} \otimes \mathbf{y}$  \textup(the bi-infinite sequence of pairs $(x_i,y_i)$ for $i\in \mathbb{Z}$\textup) generates a minimal
shift, i.e., ${\cal S}(\mathbf{x} \otimes \mathbf{y}) $  is minimal.
(b) If, in addition, the sequence $\mathbf{x}$ is computable, then  the set of patterns in ${\cal S}(\mathbf{x} \otimes \mathbf{y}) $  is also computable.
\end{lemma}

\begin{rmrk} In general, different configurations $\mathbf{x} \in {\cal T}$ in the product with one and  the same periodic $\mathbf{y}$ can result in different
shifts ${\cal S}(\mathbf{x} \otimes \mathbf{y}) $.
\end{rmrk}

In the next section we employ these lemmas in the proofs of Theorem~\ref{thm-main} and Theorem~\ref{thm-main-min}.

\subsection{Proof of Theorems~\ref{thm-main}--\ref{thm-main-min}}

\begin{proof}[The proof of Theorem~\ref{thm-main}] 
The proof of statement (b) of Theorem~\ref{thm-main} is simple. Let a one-dimensional configuration $\mathbf{x}\in {\cal A}$ 
be a projection of a two-dimensional configuration $\mathbf{y}\in{\cal B}$.  Let $w$ be a factor of $\mathbf{x}$. Then $w$ is a obtained as a projection of a vertical stripe of width $|w|$  in $\mathbf{y}$. Let us take any finite part $P$ of this stripe. Since $\mathbf{y}$ is quasiperiodic, the pattern $P$ reappears in every large enough region in $\mathbf{y}$. The projections of all these patterns result in copies of the same factor $w$ in $\mathbf{x}$. Hence, $w$ reappears in every large enough subword of $\mathbf{x}$.

Now we prove statement~(a) which is much more difficult.
In this proof we combine arguments from Sections~\ref{subsection:delegation}--\ref{subsection:delegation-quasiperiodicity} and show that the defined embedding  of a quasiperiodic one-dimensional  shift in a two-dimensional tiling  results in a quasiperiodic two-dimensional SFT.
 The main technical tools in the argument are Lemma~\ref{lemma-clones-with-embedded-bits} 
 (and Corollary~\ref{corollary-clones-with-embedded-bits}) and  Lemma~\ref{lemma-quasiperiodic-times-periodic}.

When we prove quasiperiodicity of our tiling,
we use its hierarchical structure: instead of looking for copies of all pattern $P$ in the tiling,  we  need only to show the existence of copies 
for each $(2\times 2)$-block of level-$k$ macro-tiles (for every integer $k$). To prove this property, we use Corollary~\ref{corollary-clones-with-embedded-bits}.
To explain the general scheme of the proof, we introduce a technical notion of a \emph{grid of siblings}.


\begin{definition}
\label{d:grid-of-clones}
We say that a configuration $\mathbf{x}$ \textup(a tiling\textup) contains a \emph{grid of siblings} for a $(2\times 2)$-block $P$ of level-$k$ macro-tiles, 
if in $\mathbf{x}$ there is  an infinite grid of patterns $P_{ij}$, $i,j\in\mathbb{Z}$  \textup(with a horizontal step of $q_x$ and a vertical step of $q_y$\textup) 
such that every $P_{ij}$ is also a $(2\times 2)$-block  of level-$k$ macro-tiles, and the four macro-tiles in  $P_{ij}$   play  
the same roles in the hierarchical structure 
\textup(in the sense of  Condition~\eqref{lemma-clones-with-embedded-bits-a} in Corollary~\ref{corollary-clones-with-embedded-bits}\textup) 
as  the corresponding four macro-tiles in the original pattern $P$.  

Moreover, we say that a grid of siblings is \emph{horizontally aligned} with $P$, if the grid includes a column of patterns $P_{i_0j}$, $j\in\mathbb{Z}$
that  are  aligned with $P$ in the $x$-coordinate \textup(i.e., they share with $P$  the same $(2L_k)$ columns of the tiling\textup).
\end{definition}
\begin{rmrk}
The initial block $P$ does not necessarily belong  to the grid of its siblings  (even if this grid is horizontally aligned with $P$).
\end{rmrk}
The patterns in a \emph{grid of siblings} are not necessarily equal to each other and to the original pattern $P$. Indeed, they all play similar roles
in the hierarchical structure of macro-tiles, but they may involve different parts of the embedded sequence. However, we show that in
a horizontally aligned grid of siblings there are \emph{some} patterns that are equal to the original $P$.

\begin{lemma}\label{l:sub-grid-of-clones}
If a $(2\times 2)$-block $P$ of level-$k$ macro-tiles  in our tiling has a horizontally aligned grid of siblings  $P_{ij}$
\textup(with a horizontal step of $q_x$ and a vertical step of $q_y$\textup), then 
some elements of this grid are equal to the pattern $P$. Moreover,  copies of $P$ in the grid of $P_{ij}$ can be found in every large enough 
$(M\times M)$-pattern of the tiling
\textup(the value of $M$ can be computed as a function of $k$, $q_x$, and $q_y$\textup).
\end{lemma}
\begin{proof}[Proof of lemma:] We are going to apply Corollary~\ref{corollary-clones-with-embedded-bits}. By definition, for all patterns $P_{ij}$
in the grid of siblings, we have  Condition~\eqref{lemma-clones-with-embedded-bits-a} from Corollary~\ref{corollary-clones-with-embedded-bits}.
It remains to find one $P_{ij}$ that satisfies Condition~\eqref{lemma-clones-with-embedded-bits-b} of this corollary.

To this end, we use Lemma~\ref{lemma-quasiperiodic-times-periodic}. More specifically, we focus on the  factor of the embedded sequence 
which covers the zones of responsibility of the four macro-tiles in the pattern $P$ and their fathers, and find with the help of 
Lemma~\ref{lemma-quasiperiodic-times-periodic} a copy of this factor
in some other place of the embedded sequence, with a non-zero horizontal translation divisible by $q_x$ (the horizontal step of the grid).

Since the grid of siblings is horizontally aligned with $P$, the found copy of the factor in the  embedded sequence will
be aligned with some $P_{ij}$. Due to  Corollary~\ref{corollary-clones-with-embedded-bits}, such a  pattern must be equal to $P$.

The \emph{moreover} part of Lemma~\ref{lemma-quasiperiodic-times-periodic} (see also Remark~\ref{rmrk:dense-siblings}) 
guarantees that the copies of the required factor of the embedded sequence  (taken with translations divisible by $q_x$) are dense:
they appear in every large enough segment of the embedded sequence.
Hence, we can find a block $P_{ij}$ that is equal to $P$
in every large enough part of the grid and, therefore, in every large enough $(M\times M)$-pattern  of the tiling.
(The minimal size of $M$ depends on the size of $k$, and on the steps of the grid $P_{ij}$, but not on a specific instance of $P$.)
\end{proof}
\begin{rmrk}
In the proof of the lemma we have shown a slightly stronger statement: if one pattern $P_{ij}$ in the grid is equal to $P$, then the entire column of patterns  $P_{ij'}$  in the grid consists of copies of $P$ (since all vertically aligned patterns $P_{ij'}$ share the same factor from the embedded sequence).
However,  we will only use the fact that the copies of $P$ are dense (appear in every large enough pattern in the tiling).
\end{rmrk}
In what follows we systematically apply Lemma~\ref{l:sub-grid-of-clones} to different patterns $P$. In each case, to apply the lemma we only need to 
find a \emph{horizontally aligned grid of siblings}, with uniformly bounded steps of $q_x$ and $q_y$.

\label{proof-of-thm-6-quasiperiodicity}
\smallskip
\emph{Case 1: Assume all macro-tiles in a $(2\times 2)$-block are  skeletons.}
In this case we can immediately include the given block $P$ in a grid of $(2\times 2)$-blocks of level-$k$ macro-tiles, 
with the vertical and horizontal steps of $L_{k+2}$. 
Each block in this grid has exactly the same position with respect to its father of rank $(k+1)$ and its grandfather of rank $(k+2)$. 
Thus, all blocks in this grid  are similar to each other in the sense of 
Condition~\eqref{lemma-clones-with-embedded-bits-a}, see the grid of siblings in Fig.~\ref{fig-clones-skeletons}.
Now we can apply  Lemma~\ref{l:sub-grid-of-clones}.


\begin{figure}
\centering
\includegraphics[scale=0.35]{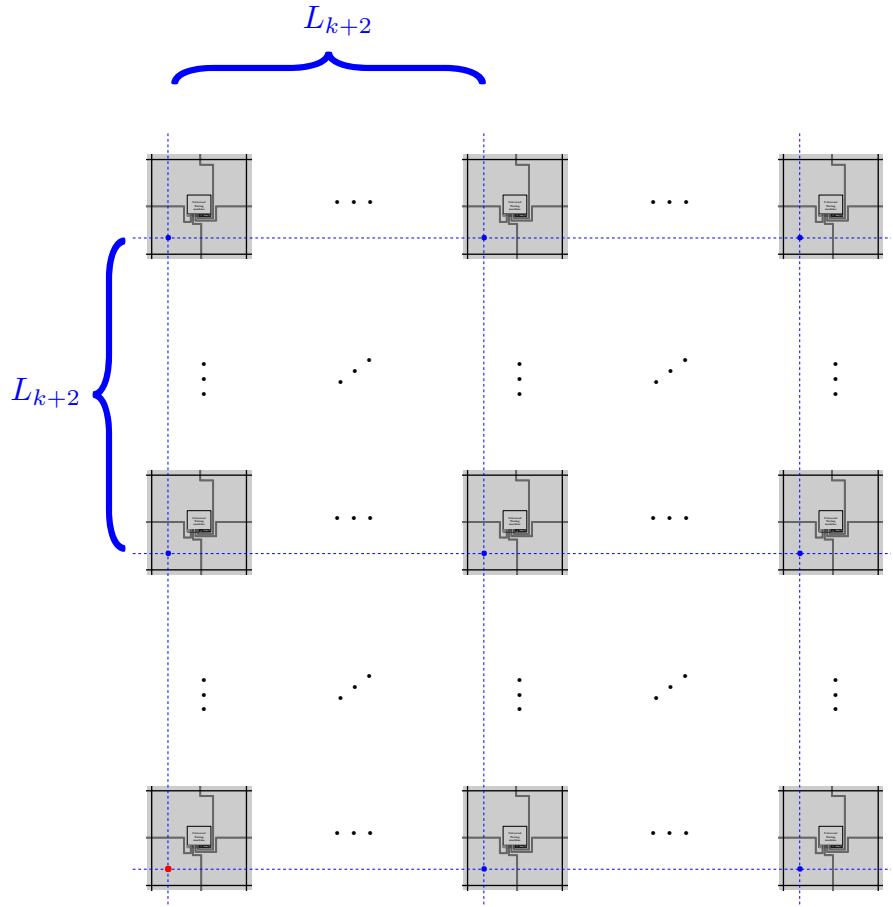} 
\caption{We are looking for a copy of a block of level-$k$ macro-tiles, which is shown as a red spot. We take its ``siblings'' --- blocks at the same position in the father level-$(k+1)$ macro-tiles  (shown as gray squares), which in turn must have the same position with respect to their level-$(k+2)$ fathers. These siblings form a regular grid (shown in blue). The step of the grid of blue siblings is equal to $L_{k+2}$, i.e., to the size of the grandfather of the initial block.}\label{fig-clones-skeletons}
\end{figure}

\smallskip
\emph{Case 2: Computation zone.}  Let us consider now the case when  the $(2\times 2)$-block of level-$k$ macro-tiles touches the computation zone.
In this case we employ the trick introduced in Section~\ref{ss-3-2}.
Recall that  for each $(2\times2)$-window that touches the computation zone,  there are only $O(1)$ admissible patterns from the space-time diagram. For each possible position of a $(2\times2)$-window in the computation zone and for each possible $(2\times 2)$-pattern in the space-time diagram, we reserved a special $(2\times2)$-\emph{diversification slot} in a macro-tile, which is essentially a block of size $2\times2$ in the ``free'' zone of the macro-tile, as shown in   Fig.~\ref{fig-macrotile-with-slots-detailed}.

Note that this diversification slot is  placed   far away from the computation zone and from all communication wires, but in the same vertical stripe as the ``original'' position of this block.  
Further, we  defined the neighbors around each  diversification slot in such a way that ``conscious memory''
(i.e., the content of input data fields (i)--(vi)) of the macro-tiles inside this slot is uniquely defined
(here we use property (p2), see Section~\ref{ss-3-2}). 

Hence, the sibling that we found for the original block of level-$k$ macro-tiles  
(the sibling which is placed in one of the diversification slots) is similar to the original block 
in the sense of the content of its computational zone and also in the sense of the involved factor of the embedded sequence.

Observe that there are infinitely many ``homologues'' of the found sibling  --- there are similar diversification slots in each level-$(k+1)$ macro-tile that take the same position with respect  to their fathers and grandfathers, see Fig.~\ref{fig-clones-comp-zone}.
All these blocks are similar to each other in the sense of Condition~\eqref{lemma-clones-with-embedded-bits-a}.
Thus, we obtain a regular grid (with the horizontal and vertical steps of $L_{k+2}$) of blocks that all satisfy Condition~\eqref{lemma-clones-with-embedded-bits-a},
and we can apply Lemma~\ref{l:sub-grid-of-clones}.
\begin{figure}
\centering
\includegraphics[scale=0.30]{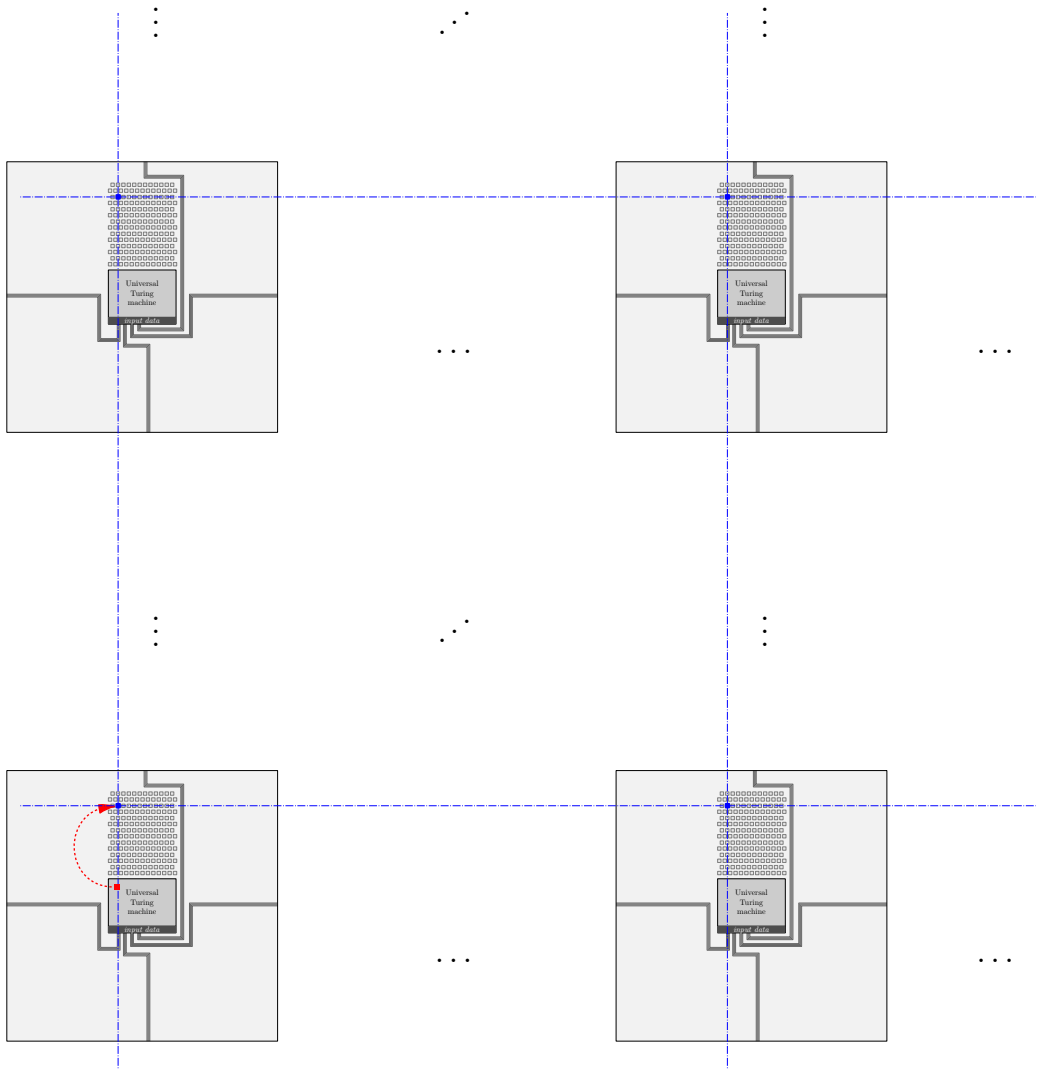}
\vspace{-10pt}
\caption{We are looking for a copy of a block of level-$k$ macro-tiles (shown as a red spot) touching the computation zone of its father. We find a periodic grid of \emph{diversification slots}  (shown as blue spots) that contain ``siblings'' of the original block in the sense of Condition~\eqref{lemma-clones-with-embedded-bits-a}. The step of this grid is $L_{k+2}$. Note that not all these siblings are similar to the original block  in the sense of Condition~\eqref{lemma-clones-with-embedded-bits-b}.}\label{fig-clones-comp-zone}
\end{figure}

\smallskip
\emph{Case 3: Communication wires.}
Now we consider the case when a $(2\times 2)$-block of level-$k$ macro-tiles involves a part of a communication wire. Due to property (p3) 
(see p.~\pageref{property:p4})
we may assume that only one wire is involved. The bit transmitted by this wire is either $0$ or $1$; in both cases, due to property (p4)
we can find another similar $(2\times 2)$-block of level-$k$ macro-tiles (at the same position within the father macro-tile of rank $(k+1)$ 
and with the same bit included in the communication wire)  in  every macro-tile of level $(k+2)$
(see the grid of siblings in Fig.~\ref{fig-clones-wires}).
\begin{figure}
\centering
\includegraphics[scale=0.3]{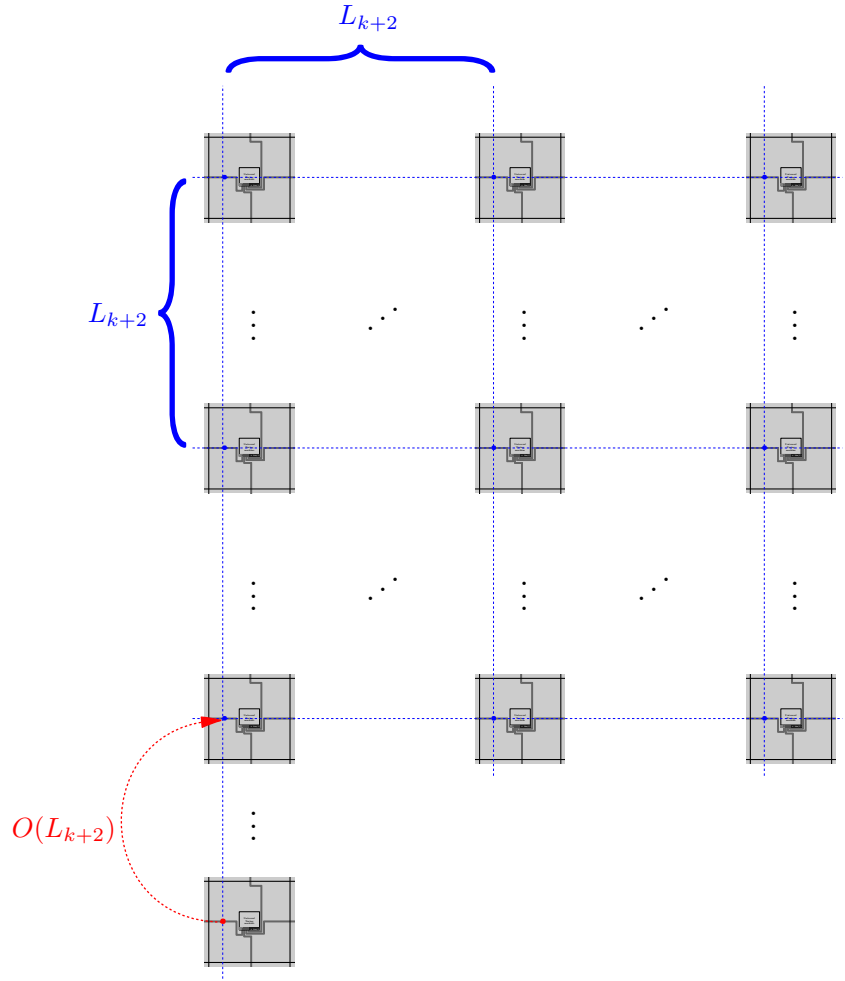} 
\caption{We are looking for a copy of a block of level-$k$ macro-tiles (shown as a red spot) touching its father's communication wire. We find a periodic grid with step $L_{k+2}$ (shown as blue spots), which consists of  ``siblings'' of the original block in the sense of Condition~\eqref{lemma-clones-with-embedded-bits-a}. Notice that the original block does not belong to this grid: it is aligned with one of the grid column but not with the rows.}\label{fig-clones-wires}
\end{figure}

We can immediately find blocks  that are similar to ours in the sense of Condition~\eqref{lemma-clones-with-embedded-bits-b},
with a vertical translation of size $O(L_{k+2})$. These blocks are obviously vertically aligned with the original block.
\begin{rmrk}
\label{r:field-vi-exception-in-use}
In this case, we find for the given block of level-$k$ macro-tiles a  sibling which has the same position with respect to their father macro-tile,
but possibly different position with respect to the grandfather. However, due to Remark~\ref{r:field-vi-exception} (p.~\pageref{r:field-vi-exception}),
for a block of macro-tiles involving a cell from a 
communication wire,  this minor displacement does not affect the field (vi) of the computation zone (which is empty). 
Therefore, the found siblings have exactly the same data in all fields (i)-(vi) as
the initial $(2\times 2)$-block of macro-tiles.
\end{rmrk}

It remains to extend the found column of siblings of $P$ and obtain a two-dimensional grid of similar siblings.
We can take the number $L_{k+2}$ as the horizontal step of the grid since
a horizontal translation of size $L_{k+2}$ does not change the position of $(2\times 2)$-blocs of level-$k$ macro-tiles with respect to their fathers and grandfathers, see Fig.~\ref{fig-clones-wires}. Then, we can again apply Lemma~\ref{l:sub-grid-of-clones}.

\smallskip
 
Thus, we have constructed a tile set $\tau$ such that every $L_k\times L_k$ pattern that appears in a $\tau$-tiling
(and which is necessarily covered by a $(2\times 2)$-block of level-$k$ macro-tiles)
 must also appear in every large enough square in this tiling. So, the constructed tile set satisfies the requirements of Theorem~\ref{thm-main}.
\end{proof}

\begin{proof}[The proof of Corollary~\ref{thm-kolmogorov}.]
To prove Corollary~\ref{thm-kolmogorov} we combine Theorem~\ref{thm-main} with a fact from  \cite{rumyantsev-ushakov}:  there exists a one-dimensional shift $\cal S$ that is quasiperiodic, and   for every configuration $\mathbf{x}\in{\cal S}$ the Kolmogorov complexity of all factors is linear, i.e.,
 $
 C(x_{i}x_{i+1}\ldots x_{i+n})  = \Omega(n)
 $
for all $i$ (here $C(w)$ denotes the Kolmogorov complexity of a string $w$).
\end{proof}

\begin{proof}[The proof of Theorem~\ref{thm-main-min}]
The proof of Theorem~\ref{thm-main-min}(b) is very similar to the proof of Theorem~\ref{thm-main}(b) and rather simple. 
We focus on the proof of Theorem~\ref{thm-main-min}(a).

First of all, we stress that the proof of Theorem~\ref{thm-main}(a) discussed above does not imply Theorem~\ref{thm-main-min}(a). If we take an effective 
 minimal one-dimensional shift $\cal A$ and plug it into the construction from the proof of  Theorem~\ref{thm-main}, we obtain a tile set  $\tau$
(simulating $\cal A$) which is quasiperiodic but not necessarily minimal. 

The property of minimality can be lost even for a periodic shift $\cal A$. Indeed, assume that the minimal period $t>0$ of the configurations in $\cal A$ is a factor of $L_k$ (i.e., of the size of each level-$k$ macro-tile in our self-simulating tiling). Then  we can extract from the resulting  SFT $\tau$ nontrivial subshifts $T_i$, $i=0,1,\ldots, t-1$ corresponding to the position of the embedded 
one-dimensional configuration with respect to the grid of macro-tiles.

To overcome this obstacle, we  superimpose some additional constraints on the embedding of the simulated $\mathbb{Z}$-shift in a $\mathbb{Z}^2$-tiling. Roughly speaking, we will enforce only ``standard'' positioning of the embedded  $1\mathrm{D}$ sequences with respect to the grid of macro-tiles. This will not change the class of the one-dimensional sequences that can be embedded in a tiling (we still get all configurations from a given minimal shift $\cal A$), but the classes of all valid tilings will reduce to some minimal SFT on $\mathbb{Z}^2$.

\emph{The standardly aligned grid of macro-tiles:} In general, the hierarchical structure of macro-tiles permits uncountably many ways of cutting the plane into macro-tiles of different ranks. We fix one particular variant of this hierarchical structure and say that a grid of macro-tiles is \emph{standardly aligned} if for each level $k$ the point $(0,0)$ is the bottom-left corner of a level-$k$ macro-tile,
see Fig.~\ref{pic:degenerate}. 
 This means that the tiling is cut into level-$k$ macro-tiles of size $L_k\times L_k$ by vertical lines with abscissae $x=L_k\cdot t'$ and ordinates  $y=L_k\cdot t''$, with $t', t''\in \mathbb{Z}$.
This structure has a kind of degeneration: the vertical line $(0,*)$ and the horizontal line $(*,0)$ serve as separating lines for macro-tiles of all ranks. 
This specific structure of macro-tiles is obviously computable.
\begin{figure}
\centering
\includegraphics[scale=0.50]{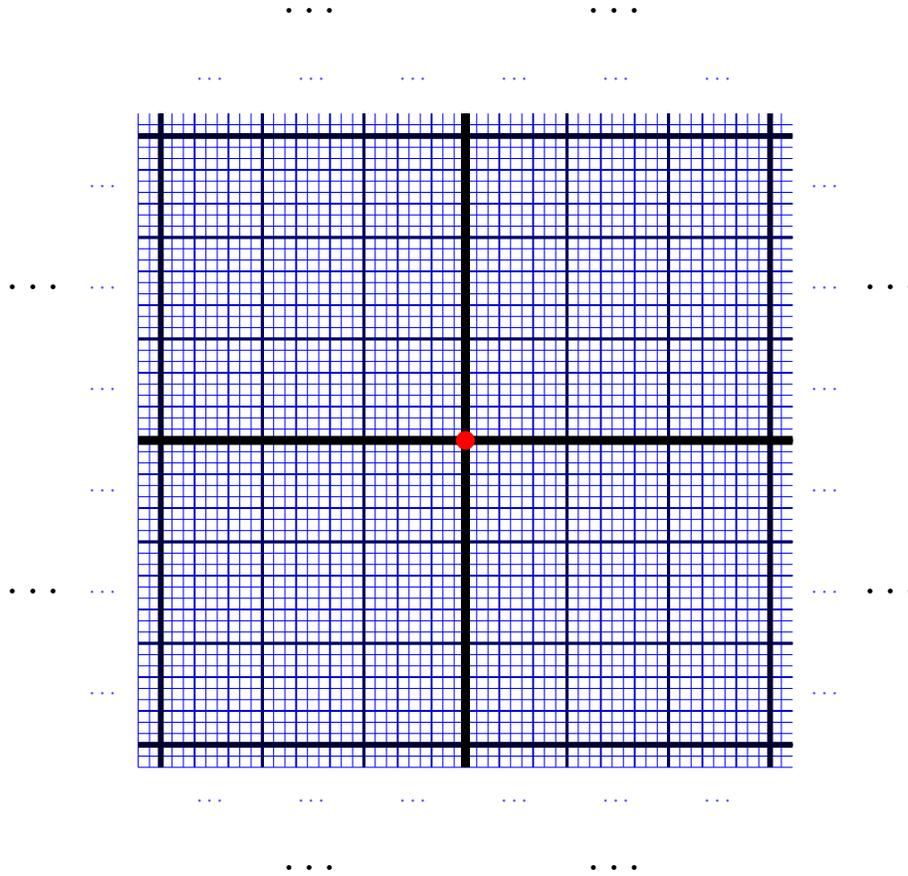}
\caption{An example of a standardly aligned grid of macro-tiles (in this example we use on each level the same zoom factor $N=3$, so every square of rank $k$  consists of $3\times 3$ squares of the previous rank). The central point (marked red) is a corner of four squares (macro-tiles) of each rank $k=1,2,3,\ldots$}\label{pic:degenerate}
\end{figure}

\emph{The canonical representative of a minimal shift:} A minimal effectively closed  $1\mathrm{D}$-shift $\cal A$ is always computable, i.e., the set of finite patterns that appear in configurations of this shift is computable, 
see \cite[Corollary~4.9]{ballier-jeandel-2008} or \cite[Proposition~9.6]{hochman}. 
 It follows immediately that $\cal A$ contains some computable configuration: we can incrementally increase a pattern maintaining the property that it 
 is globally admissible, i.e.,  appears in configurations of the shift. 
 Let us fix one (arbitrary) computable configuration $\mathbf{x}$; in what follows we call it  \emph{canonical}.

\smallskip

\emph{The standard  embedding of the canonical configuration:} We  superimpose the canonical configuration $\mathbf{x}$ on the standardly aligned grid of macro-tiles:  we take the direct product of the hierarchical structures  of the standardly aligned  grid of macro-tiles
with  the {canonical} configuration $\mathbf{x}$ from ${\cal A}$  (that is, each tile with coordinates $(i,j)$ ``contains'' in itself the letter $x_i$ from the {canonical} configuration). We split the rest of the proof into five claims.

\smallskip
 
\emph{Claim 1 (purely combinatorial): Let $m, L$ be integers and $w$ be a word. If the factor $w$ appears at least once in the standard embedding at the position $m\mod L$,  then this happens infinitely often: $w$ reappears  in the standard embedding at infinitely many positions $i$ congruent to $m\mod L$. Moreover, such positions $i$ can be found in every large enough factor of the standard embedding.}
This claim follows from Lemma~\ref{lemma-minimal-times-periodic}(a) applied to the product of the {canonical} configuration $\mathbf{x}$
with the periodic sequence 
\[
\ldots 1\ 2\ 3\ \ldots\  L\  1\ 2\ 3\ \ldots\ L\ 1\ 2\ \ldots
\]
This claim implies that the combinatorial structure of the  standard  embedding of the canonical configuration is regular:
if a factor $w$ appears at least once  in the standard embedding with a horizontal coordinate $(m\mod L_k)$ 
(where $L_k$ is the size of the level-$k$ macro-tiles),  
then $w$ reappears at the same position with respect to the grid of level-$k$ macro-tiles in every large enough pattern of this configuration. 

\begin{rmrk} \label{remark-quasi-standard}
We will say that an infinite configurations is  \emph{quasi-standard} if it contains only finite patterns that appear in the standard  embedding of the canonical configuration. The comment above applies to all quasi-standard configurations: 
if a factor $w$ appears at least once in the standard embedding at the position $m\mod L_k$,  then $w$ reappears at the same position (with respect to the grid of level-$k$ macro-tiles) in all quasi-standard configurations. On the other hand,
if a factor $w$ never appears at the position $m\mod L_k$   in the standard embedding, it cannot appear at this position in any quasi-standard configuration.
\end{rmrk}

\smallskip

\emph{Claim 2: Given a pattern $w$ of size $n\le L_k$ and an integer $i$, we can algorithmically verify whether the factor $w$ appears in the standard embedding of the canonical representative at the position $(i \mod L_k)$ relative to the grid of level-$k$ macro-tiles.} This follows from Lemma~\ref{lemma-minimal-times-periodic}(b) applied to the superposition of the canonical  configuration $\mathbf{x}$ with the  periodical grid of squares of size $L_k\times L_k$.

\begin{rmrk} \label{remark-verif}
This verification procedure is computable, but its computational complexity can be very high. 
To perform the necessary  computation, the space and time needed might be  much bigger than the length of $w$ and $L_k$.
\end{rmrk}

\emph{Upgrade of the main construction:} Before we go further, we need to update the construction of the self-simulating tiling from the proof of Theorem~\ref{thm-main}. Up to now,  we have required that every macro-tile (of every level $k$) performs in its computation zone the standard  \emph{verification procedure}, which checks whether the delegated factor of the embedded sequences contains no patterns forbidden for the shift $\cal A$. Now we make the verified property stronger: we require that the delegated factor contains only factors allowed in the shift $\cal A$, and these factors must be located at the positions (relative to the grid of macro-tiles) permitted for factors in the standard  embedding of the canonical configuration $\mathbf{x}$. This property is computable 
(due to Lemma~\ref{lemma-minimal-times-periodic}(b)), so every  forbidden pattern, or a pattern in a forbidden position, will be discovered in a computation in a macro-tile of high enough rank.

The computational complexity of this procedure can be very high (see Remark~\ref{remark-verif}), and we cannot guarantee that the forbidden patterns of small length are discovered by the computations in macro-tiles of small size. However, we do guarantee that each forbidden pattern, or a pattern in a forbidden position, is discovered by a computation in some macro-tile of high enough rank.


\smallskip

\emph{Claim 3: The new tile set admits valid  tilings of the plane.}  By construction, there exists at least one valid tiling. Indeed,  the standard  embedding of the canonical representative corresponds to a valid tiling of the plane:  in this specific tiling,  the macro-tiles of all rank never find any forbidden placement of patterns of the embedded sequence.

\begin{rmrk}
Our tile set admits many (in fact, infinitely many) different tilings. 
We cannot guarantee that all valid tilings represent the \emph{standard embedding of the canonical configuration} $\mathbf{x}$ defined above. 
However, we can guarantee a weaker property:  \emph{locally} all valid tilings look similar to each other.
More specifically, all valid configurations are \emph{quasi-standard} in the sense of Remark~\ref{remark-quasi-standard}.

\end{rmrk}

\smallskip

\emph{Claim 4: The new tile set simulates the shift $\cal A$.} This follows immediately from the construction: the embedded sequence must be a configuration without factors forbidden for~$\cal A$.

\smallskip

\emph{Claim 5: For the constructed tile set $\tau$ the set of all tilings is a minimal shift.}  We need to show that every $\tau$-tiling contains all patterns that can appear in at least one $\tau$-tiling. Similarly to the proof of  Theorem~\ref{thm-main}, it is enough to prove this property for $(2\times2)$-blocks of 
level-$k$ macro-tiles.

The argument is similar to the proof in Theorem~\ref{thm-main}. Let us fix a  $(2\times2)$-block  of level-$k$ macro-tiles that appears in a tiling $T$ and denote it $B$. We need to show that $B$ reappears in every valid tiling. That is, in every other tiling $T'$ we must find  a similar $(2\times2)$-block  of level-$k$  macro-tiles  that is identical to $B$.
The main technical tools are   Lemma~\ref{lemma-clones-with-embedded-bits} and  Corollary~\ref{corollary-clones-with-embedded-bits}.
More specifically, we need to find  in $T'$ a $(2\times 2)$-block of level-$k$ macro-tiles to which we can apply
Conditions  \eqref{lemma-clones-with-embedded-bits-a} and  \eqref{lemma-clones-with-embedded-bits-b},
see p.~\pageref{corollary-clones-with-embedded-bits}.
We will find  such a block in $T'$ in two steps.

At first we focus on Condition~\eqref{lemma-clones-with-embedded-bits-a}. 
We re-employ the argument from the proof of Theorem~\ref{thm-main} and study separately three cases:
our block of macro-tiles either appears in the ``skeletons'' area,   or it involves  part of the computation zone of the father level-$(k+1)$ macro-tile,
or it touches one of the communication wires of the father level-$(k+1)$ macro-tile. 
In each of these three cases we can find a periodic grid of blocks
with exactly the same position with respect to their fathers and grandfathers, and which are similar to the original block in the sense of Condition~\eqref{lemma-clones-with-embedded-bits-a}
(see Cases~1-3 of the proof of Theorem~\ref{thm-main}). 
In each of these cases, the $(2\times2)$-block  of level-$k$ macro-tiles that look similar to the original block $B$ in the sense of
Condition~\eqref{lemma-clones-with-embedded-bits-a} form in every tiling $T'$ a regular grid with the horizontal and vertical steps $L_{k+2}$. (Note that similar macro-tiles may appear also in other positions, outside of this regular grid; but we focus only on the regular grid of blocks with the required properties.) However,  the blocks that appear at nodes of this grid are not necessarily equal to each other: the involved macro-tiles may contain 
different factors of the embedded sequence.

We need to find in the infinite grid of $(2\times2)$-blocks a position where the block is equal to the original block $B$. 
Due to Corollary~\ref{corollary-clones-with-embedded-bits},  it is enough to find a block where we can apply 
Condition~\eqref{lemma-clones-with-embedded-bits-b}. 
Denote $w$ the factor of the embedded sequence that covers the zones of responsibility of the macro-tiles in $B$.
Let  $m$ be the horizontal coordinate of this factor with respect to the grid of level-$(k+2)$ macro-tiles. 
By construction, if  a factor $w$ of the embedded sequence may appear in a valid tiling at the position $m$ modulo $L_{k+2}$, 
than it appears  at least once (and, therefore, infinitely often, see Claim~1)  
at this specific position in the standard embedding of the canonical configuration.

Now we  use Remark~\ref{remark-quasi-standard} above. We know that the factor $w$ appears at least once at the position $m \mod L_{k+2}$ 
in  every large enough pattern of the standard embedding of the canonical configuration. Therefore,  
it appears at the same position  in every other valid configuration  of our tile set.  This concludes the proof of minimality.
 \end{proof}

\section{The proofs of the combinatorial lemmas}

\begin{proof}[Proof of Lemma~\ref{lemma-quasiperiodic-times-periodic}]
Denote by  $N$ the length of $v$.
We are given that $v$ is a factor of $\mathbf{x}$. W.l.o.g., we may assume that
 $
 v = \mathbf{x}_{[0: N-1]}.
 $
We need to prove that $v$ reappears again in $\mathbf{x}$ with a shift $t\cdot q$, i.e.,  
$
v = \mathbf{x}_{[t q : t q + N-1]}
$
for some $t>0$.

 Since $\mathbf{x}$ is recurrent, there exists an integer $l_1>0$ such that the pattern $v$ appears once again in $\mathbf{x}$ with the shift of size $l_1>l_0$ to the right,
  $$v = \mathbf{x}_{[l_1+1:l_1+N-1]}.$$ 
  If $q$ is a factor of $l_1$, then we are done.  Otherwise (if $q$ is not a factor of $l_1$), we use  the recurrency of $\mathbf{x}$ once again (now for a bigger pattern). From recurrence it follows that there exists an $l_2>0$ such that $\mathbf{x}_{[l_0:l_0+l_1+N-1]}$ appears again in $\mathbf{x}$  with some shift of size $l_2$ to the right,
  $$\mathbf{x}_{[l_0:l_0+l_1+N-1]}= \mathbf{x}_{[l_2:l_1+l_2+N-1]}.$$   
Now we have two new occurrences of $v$ in $\mathbf{x}$,
  $$
  v = \mathbf{x}_{[l_1:l_1+N-1]}= \mathbf{x}_{[l_2:l_2+N-1]}= \mathbf{x}_{[l_1+l_2:l_1+l_2+N-1]}.
  $$
If $q$ is a factor of $l_2$ or $l_1+l_2$, we get a subword in $\mathbf{x}$ (starting at the position $l_2$ or, respectively, $l_2+l_1$) that is equal to  $v$. 
Otherwise, we repeat the same argument  again, and find in $\mathbf{x}$ a copy of an even greater pattern $\mathbf{x}_{[l_0:l_0+l_1+l_2+N-1]}$. Repeating this argument $k$ times, we obtain  a sequence of positive integers  $l_1,\ldots, l_k$  such that the word $v$ reappears in $\mathbf{x}$ with all shifts composed of terms $l_i$  (all possible sums of several different $l_i$).  That is, for each integer
\begin{equation}\label{sum-l}
 \sigma = l_{j_1}+l_{j_2}+\ldots + l_{j_r}
\end{equation}
(composed of a family of $r\le k$ pairwise different $j_1,\ldots, j_r$) we have 
  $
  v= \mathbf{x}_{[\sigma:\sigma+N-1]},
  $ 
see Fig.~6.
If $k$ is large enough, then $q$ is a factor of at least one of the shifts~(\ref{sum-l}). Indeed, if $k>q(q-1)$ then we can find  $q$ different $l_j$ congruent to each other modulo $q$ (the pigeon hole principle). Then the sum of these $l_j$ must be equal to $0$ modulo $q$.  
\medskip

\begin{figure}
\includegraphics[width=\hsize]{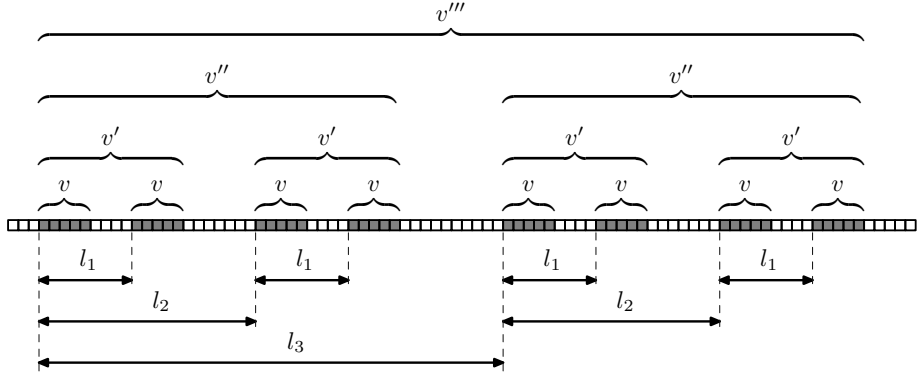}
\caption{Reappearance of factors in a quasiperiodic sequence. At  first we find a reappearance of the factor $v$, then we find a reappearance of the factor $v'$ involving two copies of $v$, then a reappearance of $v''$ involving two copies of $v'$ (and four copies of $v$), etc.}
\end{figure}

By iterating the same argument we obtain that every factor $v$ of a recurrent $\mathbf{x}$ reappears in the sequence with  infinitely many shifts divisible by $q$ (if $v=\mathbf{x}_{[i:i+n-1]}$, then there are infinitely many $t$ such that $v=\mathbf{x}_{[qt+i:qt+i+n-1]}$).

\medskip

So far  we have used only the fact that $\mathbf{x}$ is  recurrent. To prove the \emph{moreover} part of the lemma, we note that for a uniformly recurrent $\mathbf{x}$, all integers $l_j$, $j=1,\ldots,k$ in the argument above can be majorized by some  uniform  upper bound that depends only on $\mathbf{x}$ and $n$ (but not on the specific position of a factor $w$ in $\mathbf{x}$ chosen in the first place).  
This observation concludes the proof.

\end{proof}

\begin{proof}[Proof of lemma~\ref{lemma-minimal-times-periodic}]
(a) Denote by $q$ the period of $\mathbf{y}$.
Since $\cal T$ is  a minimal shift, the configuration $\mathbf{x}$ is recurrent and even quasiperiodic. By Lemma~\ref{lemma-quasiperiodic-times-periodic}, every factor $w=\mathbf{x}_{[m:n]}$ reappears in $\mathbf{x}$ with a shift divisible by $q$. Moreover, a copy of $w$ can be found with a shift divisible by $q$ in every large enough pattern
 $
 \mathbf{x}_{[i: i+L]}
 $
(where the value of $L$ depends on $w$ but not on the specific position of $m$ in $\mathbf{x}$).
It follows that the corresponding factor of the product $\tilde{w} =(\mathbf{x} \otimes \mathbf{y})_{[m:n]}$ 
reappears in every large enough pattern of all sequences in ${\cal S}(\mathbf{x} \otimes \mathbf{y})$.

\smallskip

(b) We need to verify algorithmically  whether a given factor $\tilde v$ appears in $\mathbf{x} \otimes \mathbf{y}$.  In other words, for a given word $v$ and for a given integer $i$ we need to find out whether $v$ appears in $\mathbf{x}$ in a  position congruent to $i\mod q$.

From Lemma~\ref{lemma-quasiperiodic-times-periodic} it follows that 
there exists an $L=L(v)$ such that for every appearance of $v$ in $\mathbf{x}$, 
 this factor reappears in $\mathbf{x}$  in the same position modulo $q$ with a translation at most $L$ to the right. More precisely, 
if $v=\mathbf{x}_{[i:i+|v|-1]}$ for some $i$, then there exists a positive $\sigma<L$ divisible by $q$ such that $\mathbf{x}_{[i:i+|v|-1]}=\mathbf{x}_{[i+\sigma:i+\sigma+|v|-1]}$.

Since $\cal T$ is minimal, the language of the finite factors of all configurations in $\cal T$  is computable. It follows that the bound $L=L(v)$ defined above is computable. Indeed, by the brute force search we can find the maximal possible gap between two neighboring appearances  (in this shift) of the word $v$ in  positions congruent to each other modulo $q$.
Thus, to  find out whether $v$ appears in $\mathbf{x}$ in a position congruent to $i\mod q$, it is enough to compute the first $L(v)$ letters of the sequence~$\mathbf{x}$.
\end{proof}

\bigskip
\noindent
\textbf{Acknowledgments.} 
We  thank Emmanuel Jeandel  for raising  the questions addressed in Theorem~\ref{thm-main}.
We are grateful to   Gwena\"el Richomme and Pascal Vanier   for  fruitful discussions. 
We also thank the anonymous reviewers of the ETDS journal for numerous comments and suggestions, which significantly contributed to improving the quality of the publication.

\end{document}